\documentclass[12pt]{article}
\usepackage[left=2.5cm, right=2.5cm, top=2.50cm, bottom=2.5cm]{geometry}
\usepackage{authblk,amsmath,amsfonts,amssymb,amsthm, physics}
\usepackage{algorithm}
\usepackage{algpseudocode}
\usepackage{graphicx, subfigure, tikz, makecell, multirow}
\usepackage{comment}
\usepackage{epsfig}
\usepackage{adjustbox}
\usepackage{nicefrac}
\usepackage{mathrsfs}
\usepackage{float}
\usepackage{booktabs}
\usepackage{threeparttable}
\usepackage{tabularray}
\usepackage{enumitem}  
\usepackage[utf8]{inputenc}
\usepackage[english]{babel}
\usepackage{pdfpages}
\usepackage{bbm}
\usepackage{appendix}
\usepackage{esint}
\usepackage{mathtools}
\usepackage{colortbl, color}
\usepackage{hyperref}[colorlinks,linkcolor=blue]
\numberwithin{equation}{section}

\newcommand{\e}{\mathrm{e}}
\newcommand{\E}{\mathbb{E}}
\newcommand{\R}{\mathbb{R}}
\newcommand{\bbE}{\mathbb{E}}

\newcommand{\bbR}{\mathbb{R}}
\newcommand{\bbQ}{\mathbb{Q}}

\newcommand{\de}{\mathrm{d}}
\newcommand\fa[1]{\left[#1\right]}
\newcommand\br[1]{\left(#1\right)}

\newcommand\hua[1]{\left\{#1\right\}}

\newcommand\bdTheta{\boldsymbol{\Theta}}
\newcommand\bmtheta{\boldsymbol{\theta}}

\def\calF{\mathcal{F}}

\def\calL{\mathcal{L}}

\def\calN{\mathcal{N}}
\def\calO{\mathcal{O}}

\def\calX{\mathcal{X}}

\def\bdtheta{\boldsymbol{\theta}}

\usepackage{xcolor}
\definecolor{myorange}{RGB}{201,53,56}
\hypersetup{colorlinks=True,citecolor=myorange}

\newtheorem{definition}{Definition}[section]
\newtheorem{remark}{Remark}[section]
\newtheorem{theorem}{Theorem}[section]    
   
\newtheorem{lemma}{Lemma}[section]
\newtheorem{corollary}{Corollary}[section]
\newtheorem{proposition}{Proposition}[section]

\newcommand*\xbar[1]{%
	\hbox{%
		\vbox{%
			\hrule height 0.5pt 
			\kern0.5ex
			\hbox{%
				\kern-0.1em
				\ensuremath{#1}%
				\kern-0.1em
			}%
		}%
	}%
}

\def\XXint#1#2#3{{\setbox0=\hbox{$#1{#2#3}{\int}$}
		\vcenter{\hbox{$#2#3$}}\kern-.5\wd0}}

\title{Efficient Simulation and Calibration of the Rough Bergomi Model via Wasserstein Distance}
\author{Changqing Teng \thanks{Department of Mathematics, The University of Hong Kong, Pok Fu Lam Road, Hong Kong. Email: {\tt{u3553440@connect.hku.hk}}.} \quad and \quad Guanglian Li\thanks{Corresponding author. Department of Mathematics, The University of Hong Kong, Pok Fu Lam Road, Hong Kong. Email: {\tt{lotusli@maths.hku.hk}}.}}
\begin{document}
\maketitle
\begin{abstract}
Despite the empirical success of the rough Bergomi (rBergomi) model in modeling volatility dynamics, its practical use remains challenging due to high computational complexity in both pricing and calibration arising from its non-Markovian structure. To address these difficulties, we develop an efficient computational framework. First, we propose a modified-sum-of-exponentials (mSOE) Monte Carlo scheme within the class of hybrid multifactor approximations. The method combines an exact treatment of the singular kernel over the first time step with a sum-of-exponentials approximation over the remaining time interval, and exact Gaussian simulation of the resulting multifactor components. For a fixed number of exponential terms, the method maintains linear online complexity with respect to the number of time steps. It achieves high pricing accuracy in numerical experiments, particularly for out-of-the-money options. Second, building on this pricing engine, we formulate a calibration approach based on distributional matching of the terminal underlying asset via the Wasserstein-1 distance. Instead of fitting option prices only at selected strikes, this method compares model-generated and market-implied terminal distributions through the Kantorovich-Rubinstein dual representation. Numerical experiments indicate that the mSOE scheme exhibits stable convergence, and the Wasserstein-based calibration scheme improves parameter recovery, optimization stability, and out-of-sample performance relative to conventional MSE-based fitting in the rBergomi setting considered in this paper.
	\end{abstract}

	\section{Introduction}
	Empirical studies across a wide range of financial markets suggest that log-volatility time series exhibit behaviors consistent with fractional Brownian motion (fBM) with Hurst index $H \approx 0.1$ over a broad range of time scales \cite{gatheral2022volatility}. This observation has played a central role in the development of rough volatility models. Among them, the rough Bergomi (rBergomi) model introduced in \cite{bayer2016pricing} has attracted particular attention because of its ability to jointly capture the dynamics of both historical and implied volatility.
	
	The asset price process $S$ in the rBergomi model on a filtered probability space $(\Omega, \mathcal{F}, \{\mathcal{F}_t\}_{t\in [0, T]}, \bbQ)$, with $\bbQ$ being a risk-neutral measure, satisfies the following dynamics:
	\begin{equation}\label{eq: rBergomi}
		\de S_t = rS_t\de t + S_t\sqrt{V_t}\br{\rho\de W_t + \sqrt{1 - \rho^2}\de W_t^\perp},\quad S_0 = s_0,\quad t\in[0, T],
	\end{equation}
	where $s_0$ is the initial price, constant $r$ is the risk-free interest rate, $T > 0$ is the time horizon and $W_t$ together with $W_t^\perp$ are two independent standard Brownian motions. The variance process $V_t$ satisfies
	\begin{equation}\label{eq: spot variance}
		V_t = \xi_0(t)\exp\br{\eta\sqrt{2H} \int_{0}^{t}(t-s)^{H - \frac{1}{2}}\de W_s- \frac{\eta^2}{2}t^{2H}}.
	\end{equation}
	The model dynamics are characterized by four key parameters $\bmtheta := (\xi_0(t), H, \rho, \eta)$, where the nonnegative, bounded and integrable function $\xi_0(t) := \mathbb{E}^{\mathbb{Q}}[V_t | \mathcal{F}_0]$ is the so-called initial forward variance curve. The Hurst index $H \in (0, {1}/{2})$ reflects the regularity of $V_t$ and underlies the rough volatility process. The correlation coefficient $\rho \in (-1, 0)$ and the deterministic positive constant $\eta$ is the volatility-of-volatility. \\ 
	Despite its empirical success, the practical implementation of the rBergomi model faces substantial computational hurdles. The model's reliance on the fractional kernel
	\begin{align}\label{eq: fractional kernel}
		K(t) := t^{H - \frac{1}{2}},
	\end{align}
	introduces non-Markovianity and destroys the semimartingale property, thereby invalidating conventional PDE-based pricing methods and the application of the Feynman-Kac theorem. This limitation renders Monte Carlo simulation as the primary numerical approach for option pricing under the rBergomi model. However, the simulation has significant challenges, and the central difficulty lies in the efficient discretization of the Volterra process
	\begin{equation}\label{eq: Volterra process}
		I_t := \sqrt{2H}\int_0^t K(t-s)\de W_s. 
	\end{equation}
	While exact simulation via Cholesky factorization is theoretically feasible \cite{bayer2016pricing}, its $\mathcal{O}(n^3)$ computational complexity and $\mathcal{O}(n^2)$ storage requirements render it practically infeasible for realistic applications, as summarized in Table \ref{table: MC method}. This computational bottleneck has motivated extensive research into approximate simulation schemes. Bennedsen et al. \cite{bennedsen2017hybrid} developed a hybrid scheme for Brownian semistationary processes that handles the kernel singularity by employing power function approximation near the origin and step functions elsewhere. in the case of the rBergomi model, this approach maintains exact treatment of the singularity. Subsequently, Abi Jaber and El Euch \cite{abi2019multifactor}, Bayer and Breneis \cite{bayer2023markovian}, and Zhu et al. \cite{zhu2021markovian} pursued Markovian approximations through sum-of-exponentials (SOE) representations, which transform the problem into simulating a weighted sum of diffusion processes. R{\o}mer \cite{romer2022hybrid} combined these approaches in a hybrid multifactor scheme that preserves the kernel's exact behavior for several time steps near the origin while approximating the remainder via SOE. 
	
	\begin{table}[htp]
		\centering	
		\begin{adjustbox}{max width=\textwidth}
			\begin{tabular}{l cccc}
				\toprule			
				\multirow{2}{*}{Methods} & 
				\multicolumn{2}{c}{Kernel approximation} &
				\multirow{2}{*}{Storage cost} &
				\multirow{2}{*}{Computational complexity}\\
				& Near the kernel & Remainder \\
				\hline
				Cholesky factorization & Exact & Exact & $\calO\br{n^2}$ & $\calO\br{n^3}$\\
				Hybrid  & Exact & Step function & $\calO\br{n}$ & $\calO\br{n\log n}$\\
				Multifactor & SOE & SOE & $\calO\br{n}$ & $\calO\br{Nn}$ \\
				Hybrid multifactor &  Exact & SOE & $\calO\br{n}$ & $\calO\br{Nn}$\\
				\bottomrule
			\end{tabular}
		\end{adjustbox}
		
		\caption{Comparison between the Monte Carlo-based schemes for the rBergomi model, where $n$ is the number of time steps and $N$ denotes the number of summation terms of exponentials.}
		\label{table: MC method}
	\end{table}

	Beyond pricing challenges, the calibration of the rBergomi model to market data introduces additional complexities. The standard calibration framework involves solving the following optimization problem:
	\begin{align}\label{eq: calibration}
		\inf_{\bdtheta \in \bdTheta}\frac{1}{M}\sum_{j = 1}^M\delta\br{P(\bdtheta; \zeta_j), P^{MKT}(\zeta_j)},
	\end{align}	
	where $\boldsymbol{\Theta}$ is a suitable parametric space, $\delta$ is the error metric, $\hua{\zeta_j}_{j=1}^M = \hua{(K_j, T_j)}_{j=1}^M$ is the set of contract parameters and $P(\bmtheta; \zeta_j)$ is the price of a European option under rBergomi model:
	\begin{equation}\label{eq: option price}
		P(\bmtheta; \zeta_j) =\mathbb{E}[\mathrm{e}^{-rT_j}h(S_{T_j}(\bmtheta))],
	\end{equation}
	with the payoff function $h(\cdot) = (\cdot - K_j)^+$ for the call option and $h(\cdot) = (K_j - \cdot)^+$ for the put. $P^{MKT}(\zeta_j)$ denotes the market option price. This formulation raises a critical issue: how to choose the error metric $\delta$. The traditional calibration usually uses mean squared error (MSE) and solves
	\begin{align}
		\inf_{\bdtheta \in \bdTheta}\frac{1}{M}\sum_{j = 1}^M\br{P(\bdtheta; \zeta_j)- P^{MKT}(\zeta_j)}^2.	
	\end{align} 
	This approach suffers from one fundamental limitation: it may overfit to the selected data points and perform poorly on tails of the underlying distribution that are outside of observed strikes. Furthermore, model calibration is not necessarily a convex optimization problem and often leads to multiple local minima. Liu et al. \cite{liu2019neural} presented the landscape of MSE for implied volatilities of the Heston model, with multiple local minima contained. An elaborate choice of error metric could help to mitigate this issue. In a broader context, Kidger et al. \cite{kidger2021neural} used Wasserstein-based ideas in the learning of stochastic differential equations, suggesting that distributional objectives can provide a useful alternative to pointwise fitting.

	\textbf{Our contributions}. The main contributions of this work are as follows. 
	
	First, we develop a modified sum-of-exponentials (mSOE) scheme for the rBergomi model within the class of hybrid multifactor approximations and introduce two calibration-oriented modifications. We keep the kernel exact only on the first time step, which is sufficient to preserve the local singular behavior while keeping the construction simple, and we simulate the associated Gaussian factors exactly rather than using an Euler-type discretization. We do not claim a new general kernel-approximation theory or a new hybrid principle, but combine the SOE approximation framework of Jiang et al. \cite{jiang2017fast} with this first-step exact treatment and exact Gaussian simulation. The aim here is to provide a sufficiently accurate, stable, and differentiable pricing engine for repeated use in calibration. We include a weak error analysis that connects pricing errors to the $L^\infty$-error of kernel approximation. Furthermore, we also numerically compare the proposed method with both the standard SOE method and the hybrid schemes to clarify its empirical position relative to prior approaches.
	
	Second, building on this pricing engine, we formulate a calibration approach based on distributional matching via the Wasserstein-1 distance, motivated by arguments in \cite[Section 3]{delise2021neural}. While Wasserstein-based objectives are well established in the broader literature on generative modeling and SDE-learning, to the best of our knowledge, they have not been systematically studied as calibration losses for the rBergomi model. In our setting, the calibration target is the maturity-wise risk-neutral distribution of the terminal underlying asset price. Through the Kantorovich-Rubinstein duality, the resulting objective can be interpreted as controlling pricing errors over all 1-Lipschitz payoffs rather than only over a finite set of quoted options. The emphasis here is not on a new optimal transport methodology, but on a computationally implementable adaptation of distributional calibration to the rBergomi setting.
	
	Finally, our numerical study is designed to assess whether this combination of an efficient simulation scheme and a Wasserstein-based objective improves the calibration performance within the rBergomi model, compared with conventional MSE-based fitting. We do not attempt to rank the rBergomi model against alternative stochastic models. Rather, we focus on comparisons of calibration objectives within the same model class, and examine parameter recovery, optimization behaviour, and generalization performance.
	
	The rest of the paper is organized as follows. Section \ref{sec: rBergomi simulation} details the mSOE simulation scheme and its theoretical foundations. Section \ref{sec: calibration} presents the complete calibration methodology. We illustrate in Section \ref{sec: experiment} the numerical performance of the mSOE method and the calibration scheme. Finally, Section \ref{sec: conclusion} concludes with directions for future research. Example code for the mSOE scheme and the calibration procedure is available in the GitHub repository {\url{https://github.com/evergreen1002/Calibration-rBergomi-Wasserstein-1}}.

	\section{Simulation of the rBergomi model}\label{sec: rBergomi simulation}
	In this section, we develop the modified Sum-of-Exponentials (mSOE) scheme to efficiently simulate the rBergomi model defined by \eqref{eq: rBergomi} and \eqref{eq: spot variance}. Throughout the paper, we consider an equidistant temporal grid $0 = t_0 < t_1 < \cdots < t_n = T$ with step size $\tau:= T/n$ and $t_i:= i\tau$.
	
	\subsection{Kernel approximation}
	The simulation of the Volterra process $I_t$ in \eqref{eq: Volterra process} leads to the primary computational challenge. To address this, we approximate the fractional kernel $K(t) = t^{H-1/2}, H \in (0, 1/2)$ by a sum of exponentials (SOE). We emphasize that this approximation strategy is not introduced in the present paper but is based on the constructive SOE approximation theory for power-law kernels developed in \cite{jiang2017fast}, which provides explicit uniform approximation error bounds away from the origin.
    
    To connect our setting with that framework, we begin with the Laplace transform of the power-law kernel $K(t) = t^{H-1/2}$, which is completely monotone \footnote{A function $g\br{x}$ is completely monotone if it satisfies $(-1)^kg^{(k)}(x) \geq 0$ for all $x >0$ and $k \in\mathbb{N}_{\geq 0}$.}. By Bernstein's theorem \cite{Widder:1941}, it can be expressed as:
	\begin{equation}\label{eq: Bernstein thm}
		K(t) = \frac{1}{\Gamma(1/2-H)}\int_0^\infty \mathrm{e}^{-xt}x^{-H-\frac{1}{2}}\de x=\colon\int_0^\infty \mathrm{e}^{-xt}\mu(\de x), 
	\end{equation}
	where $\mu(\de x) = w(x)\de x$, with $w(x) = \frac{1}{\Gamma(1/2-H)}x^{-H-\frac{1}{2}}$. This establishes that $K(t)$ is an infinite weighted mixture of exponentials. Applying \eqref{eq: Bernstein thm} to the Volterra process $I_t$ and invoking the stochastic Fubini theorem yields:
	\begin{align*}
		I_t = \sqrt{2H}\int_0^\infty\int_0^t \mathrm{e}^{-x(t-s)}\de W_s \mu(\de x) =\sqrt{2H}\int_0^\infty Y_t^x \mu(\de x),
	\end{align*}
	where for any $x\geq 0$, $Y^x$ is an Ornstein-Uhlenbeck process satisfying $\de Y_t^{x} = -x Y_t^{x}\de t + \de W_t$ with $Y_0^x = 0$. This formulation suggests approximating $I_t$ by discretizing the integral and using a finite sum of exponentials of the kernel:
	\begin{align}\label{eq: SOE}
		K^N(t) := \sum_{k = 1}^N \omega_k \e^{-\lambda_kt}, 
	\end{align}
	for non-negative nodes $\hua{\lambda_k}_{k=1}^N$ and weights $\hua{\omega_k}_{k=1}^N$. The relevant convergence result for this approximation is recalled below. 
		
	\begin{proposition}[{\cite[Theorem 2.1]{jiang2017fast}}]\label{prop: kernel approximation}
	For any $0 < \tau \leq T$ and any desired precision $\varepsilon > 0$, there exists a sum-of-exponentials approximation of the form
	\begin{align*}
		K^N(t) = \sum_{k=1}^N \omega_k \e^{-\lambda_k t}, \quad \omega_k, \lambda_k > 0, 
	\end{align*}
	such that 
	\begin{align*}
		\abs{K(t) - K^N(t)} \leq \varepsilon,\quad t \in [\tau, T]. 
	\end{align*}
	Moreover, the number $N$ of exponentials required to achieve this accuracy satisfies 
	\begin{align*}
		N = \calO\br{\log\frac{1}{\varepsilon}\br{\log\log\frac{1}{\varepsilon} + \log\frac{T}{\tau}} + \log\frac{1}{\tau}\br{\log\log\frac{1}{\varepsilon} + \log\frac{1}{\tau}}}.
	\end{align*}	
	\end{proposition}
Proposition \ref{prop: kernel approximation} justifies approximating the kernel by a finite exponential sum on $[\tau, T]$, while leaving the singular contribution near the origin to be treated separately. This is precisely the rationale behind our modified construction below.

Several quadrature-based constructions for power-law kernels have been proposed in the literature. As $K$ only appears inside the integrals of $I_t$, one could expect that the kernel approximation is valid in some $L^p([0, T])$ sense for $p \geq 1$. For instance, Bayer and Breneis \cite{bayer2023markovian} used Gaussian quadrature to achieve a sub-polynomial convergence rate of the $L^2$-error. Due to the better integrability of $K$ at $t = 0$ than $K^2$, they later proposed in \cite{bayer2023weak} to minimize the $L^1$-error with a faster convergence rate. However, for the calibration problems considered in the present paper, it is important that the kernel approximation permits reliable computation of gradients with respect to the Hurst index $H$. This motivates us to adopt the SOE framework in \cite{jiang2017fast}, which combines Gauss-Jacobi and Gauss-Legendre quadrature on geometric-spacing intervals. Besides providing explicit uniform approximation control on $[\tau, T]$, this construction is well-suited to stable implementation and differentiation with respect to $H$.

To further improve the accuracy, we avoid approximating the singularity at the origin by exponentials. Instead, we keep the kernel exact on the first interval and approximate only the remainder by SOE. This follows the same principle as in existing hybrid constructions \cite{bennedsen2017hybrid, romer2022hybrid}, where the local singular contribution is separated from the history part. Moreover, there is no discernible accuracy improvement to treat the kernel exactly for more than one time step, as shown in \cite{bennedsen2017hybrid, romer2022hybrid}. Accordingly, our contribution is to tailor this principle to the present calibration setting by combining a first-step exact treatment with the SOE approximation framework of \cite{jiang2017fast} and an exact Gaussian simulation of the resulting multifactor structure. As a consequence, the history terms can be updated recursively without introducing spurious cost growth. This feature helps to improve the numerical robustness of the proposed mSOE scheme.  We therefore define our modified kernel approximation as 
	\begin{equation}\label{eq: mSOE for K}
		\hat{K}(t) :=\left\{
		\begin{aligned}
			&t^{H - \frac{1}{2}}  &&t \in [t_0, t_1),\\
			&\sum_{k =1}^{N}\omega_k\mathrm{e}^{-\lambda_kt} && t \in [t_1, t_n].
		\end{aligned}
		\right.  
	\end{equation}
We stress again that the present paper does not derive a new general convergence theory for SOE approximations, nor does it propose a new hybrid principle. Instead, it combines the existing SOE approximation framework in \cite{jiang2017fast} with one-step exact treatment near the origin. The purpose of this construction is to obtain a sufficiently accurate, stable, and differentiable pricing engine that is repeatedly used in the calibration loop. The numerical comparisons in Section \ref{sec: experiment} are included to precisely clarify the empirical advantage of the construction relative to standard SOE and prior hybrid benchmarks.

\subsection{Weak error bound}
	In this subsection, we derive a weak error bound in terms of the kernel approximation error. Without loss of generality, we assume interest rate $r = 0$ and $\xi_0(t)$ is bounded on $[0, T]$, and bounded away from 0. We start by rewriting \eqref{eq: rBergomi}-\eqref{eq: spot variance} in the form
	\begin{equation}\label{eq: rBergomi 2}
		\begin{aligned}
			S_t &= s_0 + \int_0^tS_s\exp\br{X_s}\de Z_s,\quad Z_s = \rho W_s + \sqrt{1 - \rho^2}W_s^\perp,\\
			X_t &= \frac{1}{2}\log V_t  = \frac{1}{2}\br{\log\xi_0(t) + \eta I_t - \frac{\eta^2}{2}\bbE[I_t^2]}.
		\end{aligned}
	\end{equation}
	Let $\hat{S}_t$, $\hat{X}_t$ and $\hat{I}_t$ denote the corresponding processes in \eqref{eq: rBergomi 2} with $K$ replaced by $\hat{K}$. Observe that the term $t^{2H}$ in \eqref{eq: spot variance} is precisely the variance of the Volterra process $I_t$. So in the numerical scheme, we replace it by the square of the $L^2$-norm of its mSOE approximation. The following lemma, which is adapted from \cite{harms2019strong}, bounds the weak error induced by replacing $X$ with $\hat{X}$.
	\begin{lemma}[{\cite[Lemma 3]{harms2019strong}}]\label{lemma: weak error}
		Let $S, \hat{S}, X, \hat{X}: [0, T] \times \Omega \to \bbR$ be continuous stochastic processes satisfying
		\begin{align*}
			S_t = s_0 + \int_0^tS_s\exp\br{X_s}\de Z_s,\quad \hat{S}_t = s_0 + \int_0^t\hat{S}_s\exp(\hat{X}_s)\de Z_s,
		\end{align*}
		and let $f: (0, \infty) \to \bbR$ be a function such that $f\circ \exp$ is Lipschitz continuous with Lipschitz constant $L$, then there holds
		\begin{equation}\label{eq: weak error 1}
			\begin{aligned}
				&\quad\abs{\bbE\fa{f(S_T)} - \bbE[f(\hat{S}_T)]}\\ 
				&\leq L(\sqrt{T} + 6)\br{\bbE \sup_{t \in [0, T]}\br{\exp\br{2\abs{X_t}} + \exp\br{2|\hat{X}_t|}}^2}^{1/2}\br{\int_0^T \bbE[|X_t - \hat{X}_t|^2] \de t}^{1/2}. 
			\end{aligned}
		\end{equation}
	\end{lemma}	

	The weak-error bound derived below is used in this paper for European payoffs. In particular, it applies directly to put options, while the corresponding result for calls follows from put-call parity; see Remark \ref{remark: payoff}. For the sake of completeness, and to rigorously justify the applicability of the above lemma to the present rBergomi setting under the mSOE approximation, we verify below that the finite exponential moment condition is indeed satisfied. 
     \begin{proposition}
     Under the mSOE approximation in \eqref{eq: mSOE for K}, the processes $X$ and $\hat{X}$ satisfy
     \begin{align*}
     	\E\fa{\sup_{t \in [0,T]}\br{\exp(2\abs{X_t}) + \exp(2|\hat{X}_t|)}^2} < \infty. 
     \end{align*}
     \end{proposition}
     \begin{proof}
		We first note that 
		\begin{align*}
		 \sup_{t\in [0, T]}\br{\exp(2\abs{X_t}) + \exp(2|\hat{X}_t|)}^2 \leq 2\sup_{t \in [0, T]}\exp(4\abs{X_t}) + 2\sup_{t \in [0, T]}\exp(4|\hat{X}_t|).
		 \end{align*}
		Define 
		\begin{align*}
	d(t) := \frac{1}{2}\br{\log \xi_0(t) - \frac{\eta^2}{2}\E[I_t ^2]},\quad \hat{d}(t) := \frac{1}{2}\br{\log \xi_0(t) - \frac{\eta^2}{2}\E[\hat{I}_t ^2]}.
	\end{align*}
	Since $\log \xi_0$ is bounded on $[0, T]$, and both $\E[I_t^2]$ and $\E[\hat{I}_t^2]$ are uniformly bounded in $t$, there exist finite constants
\begin{align*}
D := \sup_{t \in [0, T]} \abs{d(t)} < \infty,\quad \hat{D} := \sup_{t \in [0, T]} \abs{\hat{d}(t)} < \infty. 			
\end{align*}
Moreover, we derive 
\begin{align*}
\sup_{t\in [0, T]}\exp(4\abs{X_t})\leq \e^{4D}\exp(2\abs{\eta}\sup_{t \in [0, T]}\abs{I_t}), \quad \sup_{t\in [0, T]}\exp(4\abs{\hat{X}_t}) \leq \e^{4\hat{D}}\exp(2\abs{\eta}\sup_{t \in [0, T]}\abs{\hat{I}_t}). 
\end{align*}
Since $I$ is a centered continuous Gaussian process of Riemann-Liouville fractional Brownian type on $[0, T]$, standard Gaussian-process integrability results imply that
\begin{align*}
\E\fa{\exp(2\abs{\eta}\sup_{t\in [0,T]}\abs{I_t})} < \infty. 
\end{align*}
It therefore remains to verify the corresponding property for $\hat{I}$. 	
For $0 \leq s \leq t \leq T$, the mSOE approximation yields
\begin{align*}
\hat{I}_t - \hat{I}_s = \sqrt{2H}\int_0^s (\hat{K}(t-u) - \hat{K}(s-u))\de W_u + \sqrt{2H}\int_s^t \hat{K}(t-u)\de W_u.			
\end{align*}
Then the It\^o's isometry implies
\begin{align*}
\E\left[\abs{\hat{I}_t - \hat{I}_s}^2\right] = 2H\int_0^s \abs{\hat{K}(t-u) - \hat{K}(s-u)}^2 \de u + 2H\int_s^t \abs{\hat{K}(t-u)}^2\de u. 
\end{align*}
Since $\hat{K} \in C([\tau, T])$, the second term tends to zero as $t \to s$ and the first term also tends to zero by the continuity of translations in $L^2(0, T)$. Therefore, $\hat{I}$ is mean-square continuous, and thus admits a continuous modification, which we still denote by $\hat{I}$. As $\hat{I}$ is a centered continuous Gaussian process on $[0, T]$, Fernique's theorem \cite[Theorem 2.8.5]{bogachev1998gaussian} implies that there exists $a > 0$ such that 
\begin{align*}
\E\fa{\exp(a \sup_{t \in [0, T]} \abs{\hat{I}_t}^2)} < \infty. 			
\end{align*}
For any $c, x > 0$, the inequality $cx \leq ax^2 + c^2/4a$ implies $\exp(cx) \leq \exp(ax^2)\exp(c^2/4a)$. Applying this with $c = 2\abs{\eta}, x = \sup_{t \in [0, T]} \abs{\hat{I}_t}$, we have 
		\begin{align*}
			\E\fa{\exp(2\abs{\eta} \sup_{t \in [0, T]}\abs{\hat{I}_t})} \leq \exp(\frac{\eta^2}{a})\E\fa{\exp(a \sup_{t \in [0, T]}\abs{\hat{I}_t}^2)} < \infty.
		\end{align*}		
		We complete the proof by combining the above bounds.	
	\end{proof}	
	\begin{remark}\label{remark: payoff}
		As observed in \cite[Remark 2]{harms2019strong}, the payoff of the European put option $f_p(\cdot) = (K - \cdot)^+$ satisfies the assumption in Lemma \ref{lemma: weak error} that $f \circ \exp$ is Lipschitz continuous with Lipschitz constant $K$. For call option with payoff $f_c(\cdot) = (\cdot - K)^+$, the same conclusion follows from put-call parity: 
		\begin{align*}
			\E[f_c(S_T)] = \E[f_p(S_T)] + S_0 - K,\quad \E[f_c(\hat{S}_T)] = \E[f_p(\hat{S}_T)] + \hat{S}_0 - K.
		\end{align*}
		Since $S_0 = \hat{S}_0 = s_0$ and subtraction gives 
		\begin{align*}
			\E[f_c(S_T)] - \E[f_c(\hat{S}_T)] = \E[f_p(S_T)] - \E[f_p(\hat{S}_T)]. 
		\end{align*}
		Therefore, the above lemma applies to both European call and put options.
 	\end{remark}
	
	We now derive the weak error between $S_T$ and $\hat{S}_T$ induced by approximating $K(t)$ with $\hat{K}(t)$.
	\begin{corollary}\label{cor: 1}
		Let $f: (0, \infty) \to \bbR$ be such that $f\circ \exp$ is Lipschitz continuous with Lipschitz constant $L$, then 
		\begin{align}\label{eq: weak error 2}
			\abs{\bbE\fa{f(S_T)} - \bbE[f(\hat{S}_T)]} \leq C\br{\int_{t_1}^T \abs{K(s) - \hat{K}(s)}^2 \de s}^{1/2}.
		\end{align}
		Here, the constant $C$ depends on $L$, $T$, $\xi_0(t)$, $H$ and $\eta$. 
	\end{corollary}
	\begin{proof}
		From the definition of $X_t$ in $\eqref{eq: rBergomi 2}$, we have 
		\begin{align*}
			X_t - \hat{X}_t = \frac{\eta}{2}(I_t - \hat{I}_t) - \frac{\eta^2}{4}\bbE[I_t^2 - \hat{I}_t^2].
		\end{align*}
		Writing 
		\begin{align*}
		A := \bbE[I_t^2 - \hat{I}_t^2], 
		\end{align*}
		and taking the second moments, we obtain
		\begin{align*}
			\bbE[|X_t - \hat{X}_t|^2] = \frac{\eta^2}{4}\bbE[|I_t - \hat{I}_t|^2] + \frac{\eta^4}{16}A^2.
		\end{align*}
		We estimate the two terms on the right-hand side separately. For the first term, It\^o isometry for $t \in (\tau, T]$ gives, for $t \in [\tau, T]$,
		\begin{align*}
			\mathbb{E}[|{I_t - \hat{I}_t}|^2] = 2H\int_{0}^{t -\tau}|{K(t-s) - \hat{K}(t-s)}|^2\de s\nonumber
			\leq  2H\int_{\tau}^{T}|{K(s) - \hat{K}(s)}|^2\de s. 
		\end{align*}
		For the second term, we write 
		\begin{align*}
			A = 2H\int_0^t \br{K(t-s) + \hat{K}(t-s)}(K(t-s) - \hat{K}(t-s))\de s.
		\end{align*}
		Hence, by the Cauchy-Schwarz inequality,
		\begin{align*}
			A^2 \leq 4H^2\br{\int_0^t \abs{K(t-s) + \hat{K}(t-s)}^2 \de s}\br{\int_0^t \abs{K(t-s) - \hat{K}(t-s)}^2 \de s}.
		\end{align*}
		Define 
		\begin{align}\label{eq: M}
			M := \sup_{t \in [0, T]}\int_0^t \abs{K(t-s) + \hat{K}(t-s)}^2 \de s. 		
		\end{align}
		Then 
		\begin{align*}
			A^2 \leq 4MH^2\int_\tau^T \abs{K(s) - \hat{K}(s)}^2 \de s. 
		\end{align*}
		Consequently, 		 
		\begin{align*}
			\bbE[|X_t - \hat{X}_t|^2] \leq \br{\frac{1}{2}\eta^2H + \frac{1}{4}MH^2\eta^4}\int_\tau^T|K(s) - \hat{K}(s)|^2 \de s.
		\end{align*}
		Integrating over $[0, T]$ and applying Lemma \ref{lemma: weak error} yields the desired estimate.
	\end{proof}

	\begin{corollary}\label{cor: 2}
		Suppose that the kernel approximation $\hat{K}$ in \eqref{eq: mSOE for K} satisfies
		\begin{align}\label{eq: tolerance}
			\sup_{t\in [\tau, T]} \abs{K(t) - \hat{K}(t)} \leq \varepsilon, 
		\end{align}
		for some $\tau \in (0, T)$ and $\varepsilon > 0$. Then for each fixed $H \in (0, 1/2)$, the weak error satisfies		
		\begin{align*}
			\abs{\bbE\fa{f(S_T)} - \bbE[f(\hat{S}_T)]} \leq C\sqrt{T-\tau}\varepsilon,	
		\end{align*}
		where the constant $C$ depends on $L$, $T$, $\xi_0(t)$, $H$ and $\eta$.
		\begin{proof}
			By Lemma \ref{lemma: weak error} and Corollary \ref{cor: 1}, we have 
			\begin{align*}
				&\abs{\bbE\fa{f(S_T)} - \bbE[f(\hat{S}_T)]} \\
				\leq & L(\sqrt{T} + 6)\br{\bbE \sup_{t \in [0, T]}\br{\e^{2\abs{X_t}} + \e^{2\abs{\hat{X}_t}}}^2}^{1/2}\sqrt{(\frac{1}{2}\eta^2H + \frac{1}{4}MH^2\eta^4)T}\br{\int_\tau^T|K(s) - \hat{K}(s)|^2 \de s}^{1/2},
			\end{align*}
			where $M$ is defined in \eqref{eq: M}. By the assumption \eqref{eq: tolerance}, we obtain 
			\begin{align*}
				\br{\int_\tau^T|K(s) - \hat{K}(s)|^2 \de s}^{1/2} \leq \sqrt{T-\tau}\varepsilon. 
			\end{align*}
			It remains to control the quantity $M$. Using the inequality $(a+b)^2 \leq 2a^2 + 2b^2$, $\abs{\hat{K}} \leq \abs{K} + \varepsilon$, and for sufficient accurate approximation with $\varepsilon < 1$, we derive that
			\begin{align*}
				M \leq \frac{3}{H}T^{2H} + 4T.
			\end{align*}
			Substituting this estimate into the previous inequality completes the proof.
		\end{proof}
	\end{corollary}	

\begin{remark}
	We do not claim uniformity of the hidden constants with respect to \(H\in(0,1/2)\). For each fixed $H \in (0, 1/2)$, the constant $C$ in Corollary \ref{cor: 2} is finite. Moreover, from the above proof, one could observe that the only explicit potentially singular contribution arises from $M$. Nevertheless, the weak error bound involves $M$ only through the product $MH^2$, which is of order $\calO(H)$, and no blow-up is generated by this contribution in the rough regime.
\end{remark}

	\subsection{mSOE scheme}
	We now present the complete mSOE scheme for simulating $(S_{t_i}, V_{t_i})$ on the temporal grid. The Volterra process is decomposed as
	a summation of the local part and the historical part 
	\begin{equation*}
		\begin{aligned}
			I_{t_{i+1}} &=\underbrace{\sqrt{2H}\int_{t_i}^{t_{i+1}}(t_{i+1} - s)^{H - \frac{1}{2}}\de W_s}_{I_{\mathcal{N}}(t_{i+1})} + \underbrace{\sqrt{2H}\int_{0}^{t_i}(t_{i+1} - s)^{H - \frac{1}{2}}\de W_s}_{I_{\mathcal{F}}(t_{i+1})}.		
		\end{aligned}
	\end{equation*}
	Consistent with our kernel approximation $\hat{K}$, the local part $I_{\mathcal{N}}(t_{i+1})\sim \mathcal{N}(0, \tau^{2H})$ is simulated exactly, and the historical part $I_{\mathcal{F}}(t_{i+1})$ is approximated via
	\begin{equation*}
		\bar{I}_{\mathcal{F}}(t_{i+1}) = \sqrt{2H}\sum_{k=1}^{N}\omega_k\int_{0}^{t_i}\mathrm{e}^{-\lambda_k(t_{i+1} - s)}\de W_s
		=\colon \sqrt{2H}\sum_{k=1}^{N}\omega_k\bar{I}_{\mathcal{F}}^k(t_{i+1}),
	\end{equation*}
	where the $k$th historical component $\bar{I}_{\calF}^k$ follows the recurrence:
	\begin{equation*}		
		\bar{I}_{\mathcal{F}}^k(t_{i+1}) =
		\mathrm{e}^{-\lambda_k\tau}\br{\bar{I}_{\mathcal{F}}^k(t_{i}) + \int_{t_{i-1}}^{t_i}\mathrm{e}^{-\lambda_k(t_i - s)}\de W_s},\quad i \geq 1,	
	\end{equation*}
	with initialization $\bar{I}_\calF^k(t_1) = 0$. The recurrence requires sampling a centered $(N+2)$-dimensional Gaussian vector at each step
	\begin{equation*}
		{\Xi}_i:=\br{\Delta W_{t_{i}}, \int_{t_{i-1}}^{t_i}\mathrm{e}^{-\lambda_1(t_i - s)}\de W_s, \cdots, \int_{t_{i-1}}^{t_i}\mathrm{e}^{-\lambda_{N}(t_i - s)}\de W_s, I_{\mathcal{N}}(t_{i})},
	\end{equation*}
	where $\Delta W_{t_{i}}:=W_{t_i} - W_{t_{i-1}}$. The covariance matrix $\Sigma$ of $\Xi_i$ has entries 
	\begin{equation*}
		\begin{aligned}
			&\Sigma_{1,1} = \tau,\quad\Sigma_{1, \ell} = \Sigma_{\ell, 1} = \frac{1}{\lambda_{\ell-1}}\br{1 - \mathrm{e}^{-\lambda_{\ell-1}\tau}},\quad\Sigma_{k, \ell} = \frac{1}{\lambda_{k-1}+\lambda_{\ell-1}}\br{1 - \mathrm{e}^{-(\lambda_{k-1}+\lambda_{\ell-1})\tau}},\\
			&\Sigma_{N+2, 1} = \Sigma_{1, N+2} = \frac{\sqrt{2H}}{H+1/2}\tau^{H+\frac{1}{2}},\\
			&\Sigma_{N+2, \ell} = \Sigma_{\ell, N+2} = \frac{\sqrt{2H}}{\lambda_{\ell-1}^{H + 1/2}}\gamma(H + \tfrac{1}{2}, \lambda_{\ell-1}\tau),\\
			&\Sigma_{N+2, N+2} = \tau^{2H},
		\end{aligned}
	\end{equation*}
	for $k, \ell = 2, \cdots, N + 1$, where $\gamma(\cdot, \cdot)$ refers to the lower incomplete gamma function. Since $\Sigma$ is time-independent, we only need to implement Cholesky decomposition once. As aforementioned, the term $\mathbb{E}[I_t^2] = t^{2H}$ in \eqref{eq: spot variance} is replaced by the second moment of its approximation
	\begin{align*}
		\bbE\fa{\bar{I}_{t_i}^2} &= \bbE\fa{(I_{\calN}(t_i) + \bar{I}_{\calF}(t_i))^2} 
		= \tau^{2H} + 2H\sum_{k, \ell=1}^N \frac{\omega_k\omega_\ell}{\lambda_k + \lambda_\ell}\br{\e^{-(\lambda_k + \lambda_\ell)\tau} - \e^{-(\lambda_k + \lambda_\ell)t_i}}.
	\end{align*} 
	Finally, the complete mSOE scheme proceeds as follows for $i = 0, \cdots, n-1$: 
	\begin{enumerate}
		\item Sample the Gaussian vector $\Xi_{i+1}$ using the precomputed Cholesky factor.
		\item Update historical components:
		\begin{align*}
			\bar{I}_{\mathcal{F}}^k(t_{i+1}) &= \mathrm{e}^{-\lambda_k\tau}\br{\bar{I}_{\mathcal{F}}^k(t_{i}) + \Xi_{i+1}^{(k+1)}},\quad k = 1, \cdots, N,
		\end{align*}
		where $\Xi_{i+1}^{(k+1)}$ denotes the $(k+1)$th entry of $\Xi_{i+1}$.		
		\item Compute the variance process:
		\begin{align*}
			\bar{V}_{t_{i+1}} = \xi_0(t_{i+1})\exp\br{\eta\br{\sqrt{2H}\sum_{k=1}^N\omega_k\bar{I}_\calF^k(t_{i+1}) + \Xi_{i+1}^{(N+2)}}- \frac{\eta^2}{2}\bbE\fa{\bar{I}_{t_{i+1}}^2}}.
		\end{align*}
		\item Update the asset price process by Euler-Maruyama:
		\begin{align*}
			\bar{S}_{t_{i+1}} &= \bar{S}_{t_i}\exp\br{(r - \frac{1}{2}\bar{V}_{t_i})\tau + \sqrt{\bar{V}_{t_{i}}}\br{\rho\Xi_{i+1}^{(1)} + \sqrt{1 - \rho^2}\Delta W_{t_{i+1}}^\perp}}.
		\end{align*}
		
	\end{enumerate}
	Computational complexity: The mSOE scheme requires $\mathcal{O}(N^3)$ offline cost for the Cholesky decomposition and $\mathcal{O}(Nn)$ online computation with $\mathcal{O}(N)$ storage per path.
	
	\section{Calibration scheme}\label{sec: calibration}
We now present a calibration framework for the rBergomi model based on distributional matching via the Wasserstein-1 distance. The primary calibration target is the risk-neutral distribution of the terminal underlying asset prices $S_T$ at each maturity. We begin by establishing the theoretical foundation of the Wasserstein distance and its relevance to financial model calibration.
	\subsection{Wasserstein distance}	
	\begin{definition}
		Let $(\calX, d)$ be a complete metric space. For any $p \geq 1$, the Wasserstein-$p$ distance between two Borel probability measures $\mu$ and $\nu$ on $\calX$ is defined by
		\begin{equation*}
			W_p(\mu, \nu) = \br{\inf\limits_{\gamma \in \Gamma(\mu, \nu)}\mathbb{E}_{(x,y)\sim \gamma}d^p(x, y)}^{1/p},
		\end{equation*}
		where $\Gamma(\mu, \nu)$ is the set of all Borel probability measures on $\calX \times \calX$ with marginals $\mu$ and $\nu$. 
	\end{definition}
	Intuitively, the Wasserstein distance measures the minimum work required to transform one probability distribution into another, where the cost of moving mass is proportional to the distance raised to the power $p$. We specialize to the case $p = 1$ with $\mathcal{X} = \bbR$ equipped with the Euclidean metric $d(x,y) = \abs{x-y}$. The following theorem shows that for real-valued random variables, the Wasserstein-1 distance admits a tractable representation that bypasses the need for solving the optimal transport problem directly.
	
	\begin{theorem}[{\cite[Theorem 2.1]{de20211}}]
		Let $X$ and $Y$ be real-valued random variables with cumulative distribution functions $F_X$ and $F_Y$, respectively. The Wasserstein-1 distance between $X$ and $Y$ equals the area between their distribution functions:
		\begin{align*}
			W_1(X, Y) = \int_\bbR \abs{F_X(u) - F_Y(u)}\de u = \int_{0}^1 \abs{F_X^{-1}(z) - F_Y^{-1}(z)}\de z,
		\end{align*}
		where the quantile function $F^{-1}: [0, 1] \to \bbR$ is defined by 
		\begin{align*}
			F^{-1}(z) := \inf\hua{u\in \bbR: F(u) \geq z}.
		\end{align*}
	\end{theorem}
	The quantile function formulation informs the computational approach. In practical applications, we work with empirical distributions derived from samples.  
	\begin{definition}\label{def: empirical_wass}
		Given $m$ samples from each distribution, $\hua{X_i}_{i=1}^m$ and $\hua{Y_i}_{i=1}^m$, the empirical Wasserstein-1 distance is defined by 
		\begin{equation}\label{eq: empirical W_1}
			W_1(\hua{X_i}, \hua{Y_i}) = \frac{1}{m}\sum_{i = 1}^m|{X_{(i)} - Y_{(i)}}|,
		\end{equation}
		where $X_{(i)}$ and $Y_{(i)}$ denote the $i$-th order statistics of the samples.
	\end{definition}
	The most profound insight for our calibration methodology comes from an alternative characterization of the Wasserstein-1 distance provided by Kantorovich-Rubinstein duality \cite{villani2021topics}.
	\begin{definition}\label{def: KR duality}
		Let $X, Y$ be two random variables, then the Kantorovich-Rubinstein duality states that 
		\begin{equation}\label{eq: KR duality}
			W_1(X, Y) = \sup\limits_{\mathrm{Lip}(f) \leq 1}\br{\bbE\fa{f(X)} - \bbE\fa{f(Y)}},
		\end{equation}
		where $\mathrm{Lip}(f)$ denotes the Lipschitz constant of $f$ and $f: \bbR \to \bbR$.
	\end{definition}
	This duality result transforms the primal definition into a supremum over a class of test functions, which provides a powerful adversarial interpretation to the calibration approach.

	\subsection{Calibration via Distributional Matching}
	We present the calibration approach in this section, and the main aim is to match the maturity-wise risk-neutral distribution of the terminal underlying asset price. For each maturity $T_j$, let $S_{T_j}^{MKT}$ denote the market-implied terminal asset price under the risk-neutral measure, and let $S_{T_j}(\bmtheta)$ denote the corresponding model-generated terminal asset price under parameters $\bmtheta$.	Suppose there exists a desirable set of model parameters $\bmtheta_{\text{true}}$ for the rBergomi model so that the model price perfectly matches the market price, i.e.,
	\begin{align*}
		P\br{\bmtheta_{\text{true}}; K, T} = P^{MKT}\br{K, T},\quad \text{for all}\ K\text{ and } T.
	\end{align*} 
	However, such an identity may fail in practice for two reasons. On the one hand, every model is inevitably misspecified to some extent, so the observed option prices may not lie exactly in the range of prices generated by the model. On the other hand, observed option prices are noisy and only defined up to bid-ask spreads. Conventional model calibration typically minimizes pointwise discrepancies on a finite set of contracts, which generally suffers from overfitting. Our framework instead calibrates the model by matching the implied terminal distributions. 
	
	Consider a European option with payoff $h$, which is a 1-Lipschitz function of the underlying price. For simplicity, assume a zero interest rate. Then, by the Kantorovich-Rubinstein duality \eqref{eq: KR duality}, the pricing error satisfies 
	\begin{align}\label{eq: inequality}
		\abs{P\br{\bmtheta; K, T} - P^{MKT}\br{K, T}}  
		&= \abs{\bbE\fa{h\br{S_T\br{\bmtheta}}} - \bbE\fa{h\br{S_T^{MKT}}}} \nonumber \\
		&\leq W_1\br{S_T\br{\bmtheta},S_T^{MKT}}\quad \text{for all}\ K.
	\end{align}
	This indicates that the Wasserstein-1 distance between the model-generated and market-implied distributions of the terminal underlying asset provides a uniform control on pricing errors over all 1-Lipschitz payoffs. The duality perspective offers an adversarial interpretation: the Wasserstein-1 distance corresponds to the worst-case pricing error over all 1-Lipschitz payoffs, and implicitly identifies the strike that maximizes the discrepancy between model and market prices. This leads to our proposed calibration objective that minimizes the Wasserstein-1 distance between the model-generated and market-implied distributions across all maturities. 
	
	We define the main calibration function by 
	\begin{align}\label{eq: loss}
		\calL\br{\bmtheta} = \frac{1}{M}\sum_{j=1}^MW_1(S_{T_j}\br{\bmtheta}, S_{T_j}^{MKT}).
	\end{align}
	 To describe the calibration problem in a general form, we allow flexible parameterizations of the model inputs $(\xi_0(t), H, \rho, \eta)$. We write
	\begin{align*}
		\xi_0(t) := \xi_0(t;\boldsymbol{\theta}_\xi),\quad H := \theta_H,\quad \rho := \theta_\rho,\quad \eta := \theta_\eta,
	\end{align*}
	and collect all unknown quantities into the finite-dimensional parameter vector $\bmtheta := (\bmtheta_\xi, \theta_H, \theta_\rho, \theta_\eta)$, where $\bmtheta_\xi$ denotes the parameters used to represent $\xi_0(\cdot)$, while $H, \rho$ and $\eta$ are treated as directly trainable scalar variables. In the numerical implementation used in this paper, the resulting finite-dimensional optimization problem is solved by L-BFGS-B under box constraints. The optimization is terminated according to the standard stopping criteria of L-BFGS-B, namely when the projected gradient norm falls below a prescribed tolerance, when the relative decrease of the objective falls below a prescribed tolerance, or when the maximum number of iterations is reached. 
	
	The complete calibration procedure is summarized in Algorithm \ref{alg: Wasserstein calibration}. The algorithm compares one-dimensional terminal distributions independently across maturities, rather than matching full path-space distributions. It should be interpreted as a general calibration workflow once the maturity-wise target distributions have been specified. In a synthetic-data setting, the target empirical distributions can be generated directly from prescribed ground-truth parameters. In practical applications, the target distributions are not directly observable and must be recovered from option prices, see \cite{figlewski2018risk} and references therein for a review of this direction. This recovery step is addressed in Section \ref{subsec: RND extraction} through stochastic volatility inspired (SVI) parameterization, followed by numerical inversion of the recovered cumulative distribution function (cdf). After the preprocessing step, the resulting target quantiles can be fed into Algorithm \ref{alg: Wasserstein calibration} exactly the same way as in the synthetic setting. 
	
	\begin{algorithm}	
		\caption{Calibration of rBergomi model by Wasserstein-1 distance}
		\label{alg: Wasserstein calibration}
		\begin{algorithmic}	
			\Require{Time grid $0 = t_0 < t_1 < \cdots < t_n = T$; \\
				Maturity set $\hua{T_j}_{j = 1}^M$;\\
				Target sorted empirical samples $\{S_{T_j}^{MKT, (i)}\}_{i=1}^m$, arranged in ascending order for each maturity $T_j$; \\
			    Initial parameter vector $\bmtheta^{(0)} = (\bmtheta_\xi, \theta_H, \theta_\rho, \theta_\eta)$; \\
				Box-constrained parameter domain $\boldsymbol{\Theta}$;\\
				Number of Monte Carlo samples per maturity at each iteration $m$;\\
				Maximum number of iterations $k_{\text{max}}$;\\
				Stopping tolerance $\epsilon_f, \epsilon_g$.
			}	          
			\State{\textbf{Initialization:}
				$\bmtheta \gets \bmtheta^{(0)}$}			
			\For{$k = 0, 1, \cdots, k_{\text{max}}-1$}
			\For{$T_j$ in $\hua{T_1, \cdots, T_M}$}		
			\State{Generate $m$ Monte Carlo samples $\{S_{T_j}^{(i)}(\bmtheta^{(k)})\}_{i=1}^m$ by the mSOE scheme;}
			\State{Sort the simulated samples in ascending order:
			\begin{align*}
			S_{T_j}^{(1)}(\bmtheta^{(k)}) \leq \cdots  \leq S_{T_j}^{(m)}(\bmtheta^{(k)});
			\end{align*}}
	        \State{Compute the empirical Wasserstein-1 distance:
	        \begin{align*}
	        W_1^{(j)}(\bmtheta^{(k)}) = \frac{1}{m}\sum_{i=1}^m\abs{S_{T_j}^{(i)}(\bmtheta^{(k)}) - S_{T_j}^{MKT, (i)}};	        
	        \end{align*}}
			\EndFor	
			\State{Compute the calibration loss:
			\begin{align*}
			\calL(\bmtheta^{(k)}) = \frac{1}{M}\sum_{j=1}^M W_1^{(j)}(\bmtheta^{(k)});		
			\end{align*}}	
		    \State{Compute the gradient $\nabla_{\bmtheta} \mathcal{L}(\bmtheta^{(k)})$ by automatic differentiation;}
		    \State{Perform one box-constrained L-BFGS-B update for $\bmtheta^{(k)}$ subject to the box constraints $\boldsymbol{\Theta}$;} 
	        \State{Check stopping criteria:}
	        \If{the projected gradient norm is below $\epsilon_g$ or if the relative decrease of the objective is below $\epsilon_f$}
	        \State{Terminate the calibration;}
	        \EndIf
	        \EndFor
			\State \Return{calibrated parameter set $\bmtheta^\ast$.}
		\end{algorithmic}
	\end{algorithm}

	Within the calibration setting considered in this paper, the Wasserstein-based objective presents the following potential advantages relative to conventional pointwise price-based fitting for the same model:
	\begin{itemize}
		
		\item Distributional Alignment: By matching the entire risk-neutral distribution rather than prices at discrete strikes, our method ensures that the calibrated model accurately captures the market-implied dynamics across all strike levels. This is particularly crucial for pricing exotic options whose values depend on the full distribution or path properties, not just terminal payoffs at specific strikes. 
		
		\item Robustness to Overfitting: The uniform bound interpretation provided by the Kantorovich-Rubinstein duality ensures that our calibration objective controls pricing errors across all possible strikes, not just those included in the calibration set. This inherent regularization makes the method particularly robust and reduces overfitting to potentially noisy price observations.
		
		\item Computational Tractability through Monte Carlo: The "simulate-and-compare" approach enabled by our efficient mSOE scheme, circumvents the need for closed-form pricing formulas. This suggests that the same “simulate-and-compare” principle may be applicable to a broader class of complex models for which efficient Monte Carlo simulation is available. The empirical Wasserstein distance between one-dimensional distributions is computationally efficient, which only requires $\calO(m\log(m))$ operations per evaluation.
		
		\item Improved Optimization Behavior: As we demonstrate empirically in Section \ref{subsubsec: loss landscape}, the Wasserstein-1 objective yields smoother and less irregular optimization landscapes than traditional metrics such as MSE in the experiments considered in \ref{subsubsec: calibration_to_European_options}. This improves the convergence stability, reduces sensitivity to initialization, and enhances the reliability of the calibration process.
	\end{itemize}

	\section{Numerical tests}\label{sec: experiment}
	This section presents a comprehensive numerical evaluation of the proposed mSOE scheme and the Wasserstein-1 distance-based calibration algorithm. We first assess the numerical accuracy and convergence properties of the mSOE scheme against the SOE scheme and the hybrid benchmark, examining both implied volatility smiles and surfaces. Subsequently, we demonstrate the performance of the calibration framework, highlighting its generalization capability and parameter identifiability compared to the conventional MSE approach. Unless otherwise stated, benchmark quantities of this section are generated independently by Cholesky factorization, whereas candidate model evaluations during calibration are performed by the mSOE scheme.
	\subsection{Numerical accuracy of the mSOE scheme}
	
	\subsubsection{Kernel approximation}
	To complement the theoretical discussion, we numerically examine the convergence of the mSOE approximation of the kernel. The model parameters are specified as follows \cite{bennedsen2017hybrid} 
	\begin{align}
	\label{eq: rBergomi param}
	S_0 = 1,\quad \xi_0(t) = 0.235^2\quad H = 0.07, \quad \rho = -0.9,\quad \eta = 1.9,
	\end{align}
	We define the kernel approximation error $\varepsilon_N$ and the option pricing error as $\mathcal{E}_N$ as 
	\begin{align*}
		\varepsilon_N := \max_{t_i \in [t_1, T]} \abs{K(t_i) - \hat{K}(t_i)},\quad \mathcal{E}_N = \abs{\E[f(S_T)] - \E[f(\hat{S}_T)]},
	\end{align*}		
	As shown in the left panel of Figure \ref{fig: mSOE_convergence}, the approximation error decreases rapidly as $N$ increases for all reported values of $H$. The right panel shows the relation between the weak pricing error $\mathcal{E}_N$ and the kernel approximation error $\varepsilon_N$ on the intermediate range $N \in \hua{8, 12, 16, 20}$. For larger values of $N$, the pricing error is no longer dominated by the kernel approximation but affected by time-discretization and Monte Carlo noises, so it ceases to decrease monotonically with $\varepsilon_N$. On this range, the fitted slopes are $0.931$ for $H = 0.1$ and 0.851 for $H = 0.2$, which are reasonably close to the theoretical benchmark of 1. For the rougher case $H = 0.05$, the fitted slope is 0.382, which suggests that the asymptotic regime is harder to observe numerically in the highly rough setting. 
	\begin{figure}[hbt!]
		\centering
		\subfigure{
			\includegraphics[width=0.95\linewidth ,clip]{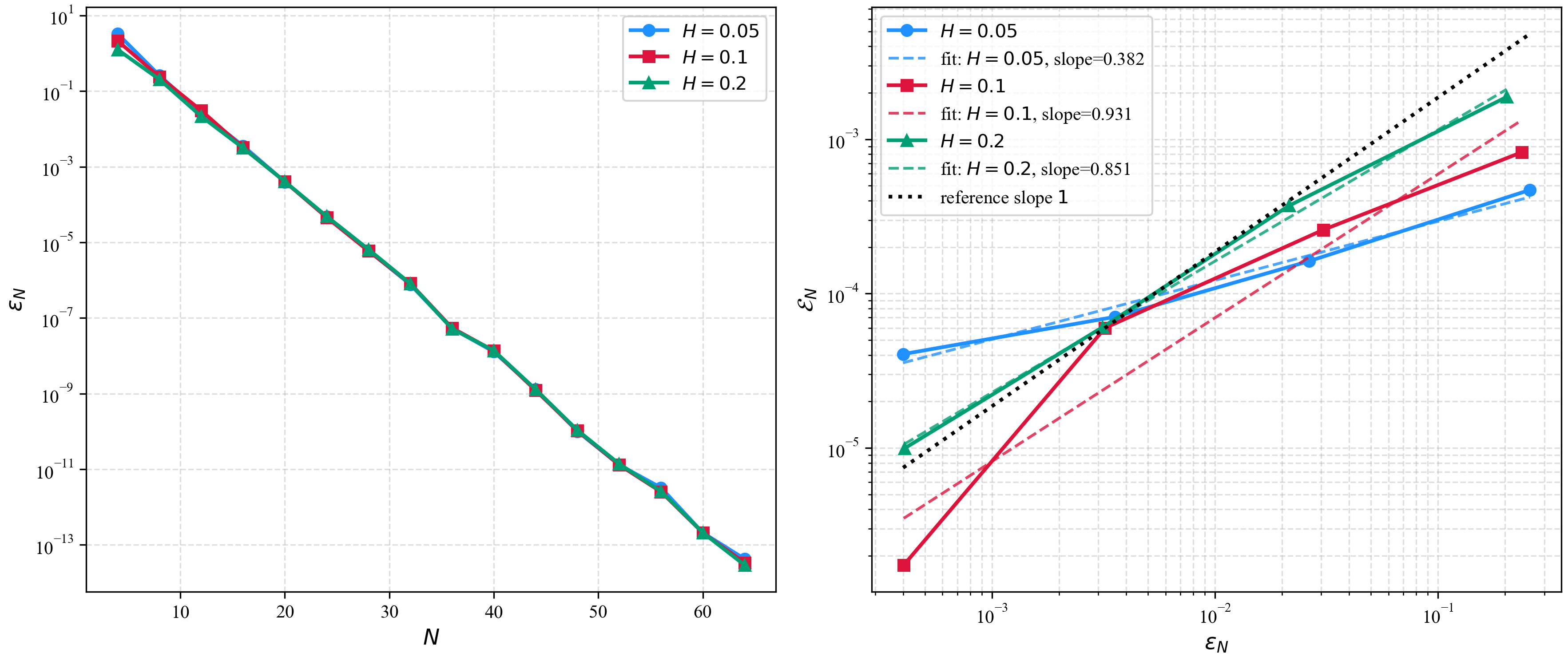}}	
		\caption{Convergence test for the kernel approximation \eqref{eq: mSOE for K} with $T = 1, K = 1, \tau = 1/1024$ and $2^{24}$ Monte Carlo paths. Reference simulations are obtained by Cholesky factorization. Left: kernel approximation error $\varepsilon_N$ versus number of summation terms $N$ for $H = 0.05, 0.1, 0.2$. Right: weak pricing error $\mathcal{E}_N$ with $f(\cdot) = (\cdot - K)^+$ versus $\varepsilon_N$ on a log-log scale. The dashed lines are least-squares fits over $N = \hua{8, 12, 16, 20}$; the dotted black line indicates slope 1. The fitted slopes are 0.382, 0.931, and 0.851 for $H = 0.05, 0.1, 0.2$, respectively.}
		\label{fig: mSOE_convergence}
	\end{figure}

	\subsubsection{European call option}	
	To better position the mSOE scheme relative to the existing literature, we compare it not only with the SOE method, which is detailed in \cite[Section 4.3]{bayer2023markovian}, but also with the hybrid benchmark of Bennedsen et al. \cite{bennedsen2017hybrid}. We adopt the same model parameters as in \eqref{eq: rBergomi param} with a maturity $T = 1$ and a set of 21 log-strikes $k_i = -0.55 + 0.05 \times i$ for $i = 1, \cdots, 21$. The benchmark implied volatility smile is computed via the Cholesky factorization. For both the SOE and mSOE schemes, the nodes $\hua{\lambda_k}$ and weights $\hua{\omega_k}$ for the sum-of-exponentials approximation are determined using the method introduced in \cite[Section 2]{jiang2017fast}. Their specific values for varying numbers of summation terms $N$ are provided in Table \ref{tab: nodes and weights}. An additional objective is to numerically investigate the convergence rates of these schemes as the number of steps $n$ increases.

	Figure \ref{fig: imp vol for European call option} displays the implied volatility smiles generated by the Cholesky factorization, hybrid, SOE, and mSOE schemes for a fixed number of time steps $n = 128$. The SOE scheme exhibits minor deviations from the benchmark for at-the-money strikes and significant inaccuracies for out-of-the-money strikes. By contrast, both the hybrid and mSOE schemes mitigate these discrepancies substantially. When $N = 16$, the mSOE smile becomes visually indistinguishable from the benchmark and performs comparably to the hybrid approximation.

	To numerically determine the convergence rates of the implied volatility smiles, we conduct simulations with a sufficiently large sample size of $m = 2^{24}$ to minimize Monte Carlo error. Table \ref{tab: maximal rele err for smiles} reports the maximal relative errors over the strike range for both hybrid, SOE, and mSOE schemes, with the results visualized in Figure \ref{fig: weak convergence rate}. The SOE scheme fails to exhibit a clear convergent pattern as $n$ increases. The hybrid scheme is clearly more stable, which confirms the benefit of treating the singularity exactly. The mSOE scheme also demonstrates stable convergence across different values of $N$. For sufficiently large $N$, it attains comparable or even smaller errors than those of the hybrid scheme. As illustrated in Figure \ref{fig: weak convergence rate with varying N}, the convergence rate of the mSOE scheme improves as the number of summation terms $N$ increases.
	
	\begin{figure}[hbt!]
		\centering
		\subfigure{
			\includegraphics[width=1\linewidth ,clip]{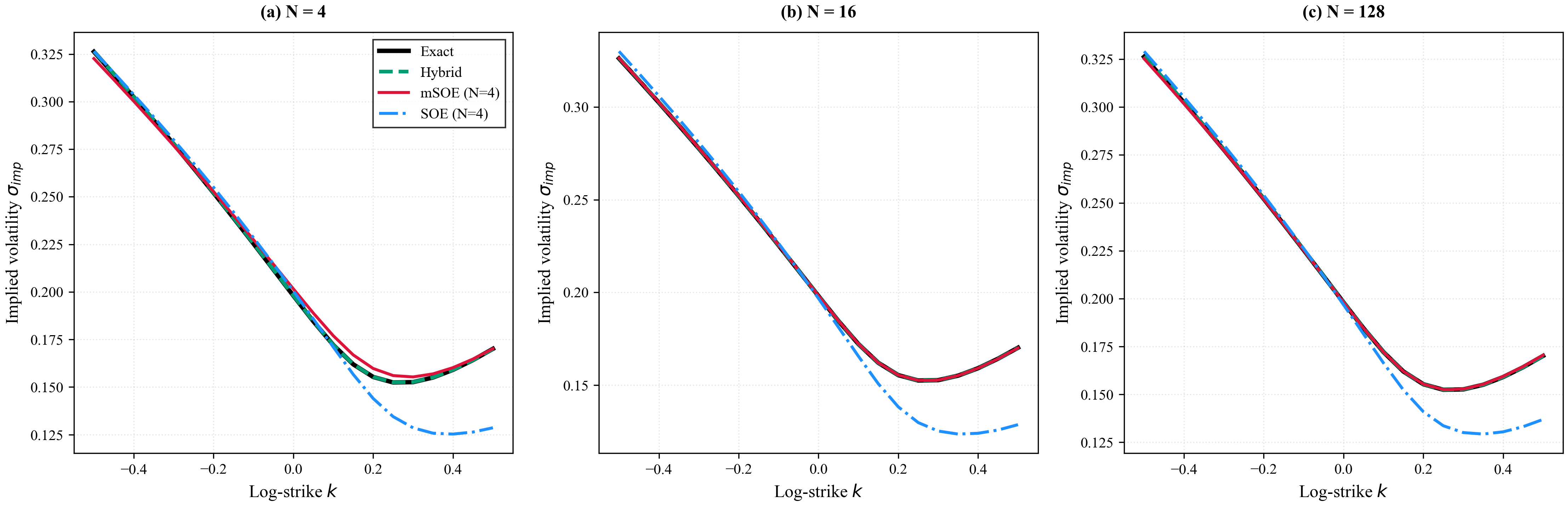}}	
		\caption{Implied volatility smiles using $2^{24}$ samples with  time steps $n = 128$, $H = 0.07, T=1$ using nodes and weights given in Table \ref{tab: nodes and weights} for $N = 4$ (left), $N = 16$ (middle) and $N = 128$ (right). }
		\label{fig: imp vol for European call option}
	\end{figure}
	
	\begin{table}[htp]
		\centering
		\begin{adjustbox}{max width=\textwidth}
			\begin{tabular}{c|c| ccccc ccccc}
				\toprule
				\multirow{2}{*}{$N = 4$} & $\lambda$&  0.047108& 1& 21.227784&   450.618823&&&&\\
				\multirow{2}{*}{} & $\omega$ &0.398569 & 1.482765& 5.516218&    20.521561&&&& \\
				\hline
				\multirow{2}{*}{$N = 8$} & $\lambda$ &0.161212 & 0.025989& 0.004190& 1& 6.203015& 38.477401& 238.675916&  1480.51040\\
				\multirow{2}{*}{} & $\omega$ & 0.404082& 0.184353& 0.084107& 0.885703& 1.941367& 4.255266 & 9.327086 & 20.443970\\
				\hline
				\multirow{4}{*}{$N = 16$} & $\lambda$  &0.470200& 0.221072&  0.102817& 0.039206& 0.004176& 1 &  2.126753 & 4.523076\\
				\multirow{4}{*}{} &$\omega$& 0.264736& 0.191490& 0.145476& 0.152321& 0.200620&  0.366211 & 0.506582 & 0.700758 \\			
				\multirow{4}{*}{} & $\lambda$  &9.619464& 20.458220& 43.509572& 92.534093& 196.797117& 418.538770& 890.128395 & 1893.08283\\
				\multirow{4}{*}{} & $\omega$ &0.969363& 1.340926 & 1.854911& 2.565910& 3.549438& 4.909960& 6.791977 & 9.395384 \\			
				\hline
			\end{tabular}
		\end{adjustbox}
		\caption{Nodes and weights with varying $N$s with $H = 0.07$, $\tau = 1/128$ and $T = 1$.}
		\label{tab: nodes and weights}
	\end{table}

	\begin{table}[htp]
		\centering	
		\begin{adjustbox}{max width=\textwidth}
			\begin{tabular}{c|c|cc|cc|cc}
				\toprule	
				& &\multicolumn{2}{c}{$N = 4$} & \multicolumn{2}{c}{$N = 16$} & \multicolumn{2}{c}{$N = 128$} \\
				\hline 
				\makecell{$n$} & Hybrid & SOE & mSOE & SOE & mSOE & SOE & mSOE \\
				\hline 			           
				$128$ &0.1607 &0.1224& 0.1617&0.1277& 0.1604& \textbf{0.0841}& 0.1638\\
				$256$ &0.0950&0.1685& 0.0930& 0.1737& \textbf{0.0920}&0.1079& 0.0957\\
				$512$ &0.0473&0.2042& 0.0473&0.1842& 0.0451&0.1227& \textbf{0.0443}\\
				$1024$ & 0.0177&0.1448& 0.0241& 0.1687& \textbf{0.0169}&0.1278& 0.0172\\
				$2048$ & 0.0016&0.1564& 0.0148 &0.1471& 0.0025&0.1248& \textbf{0.0009}\\
				\bottomrule
			\end{tabular}
		\end{adjustbox}
		\caption{Maximal relative error (over strikes) for the implied volatility smiles with $H = 0.07$, $T = 1$ and $m = 2^{24}$.} 
		\label{tab: maximal rele err for smiles}
	\end{table}
	
		\begin{figure}[hbt!]
		\centering
		\subfigure{
			\includegraphics[width=1\linewidth ,clip]{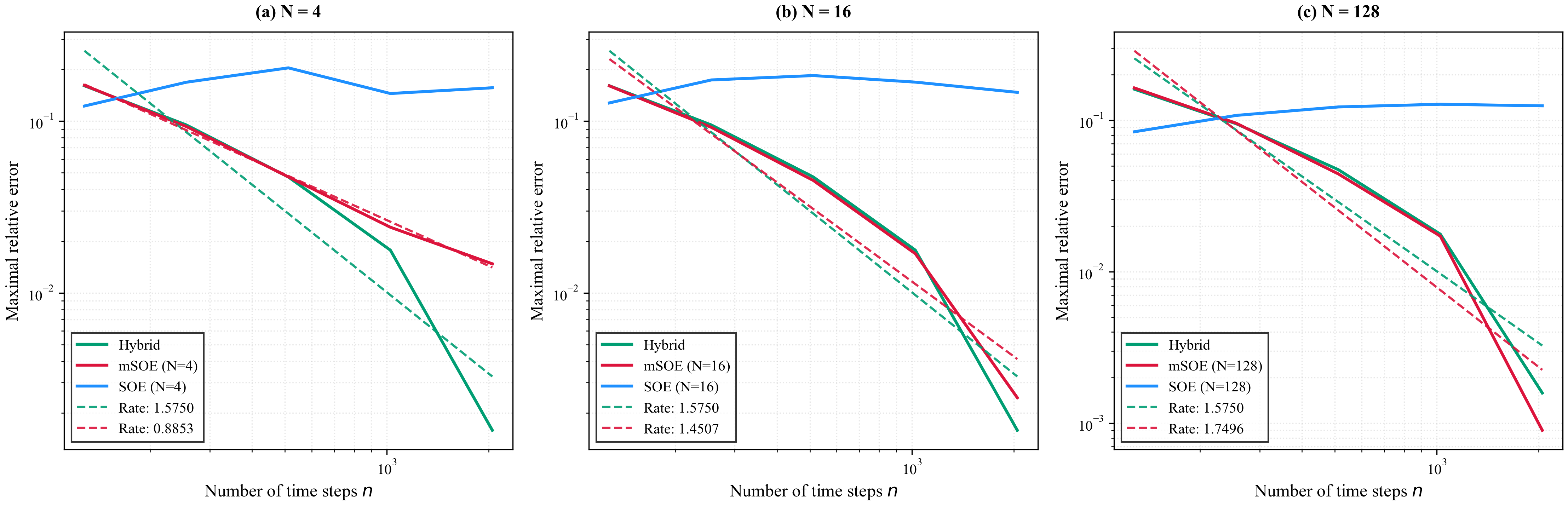}}	
		\caption{Maximal relative errors of implied volatility smiles using $m = 2^{24}$ samples with $H = 0.07$, $T = 1$ and $N = 4$ (left), $N = 16$ (middle) and $N = 128$ (right). The red, blue, and green curves represent the maximal errors over strikes between the true smile and the smile obtained by the mSOE, SOE, and hybrid schemes, respectively. The dashed lines give a reference for the weak convergence rates. }
		\label{fig: weak convergence rate}
	\end{figure}

	\begin{figure}[hbt!]
		\centering
		\subfigtopskip = 2pt
		\subfigure{
			\includegraphics[width=0.6\linewidth ,clip]{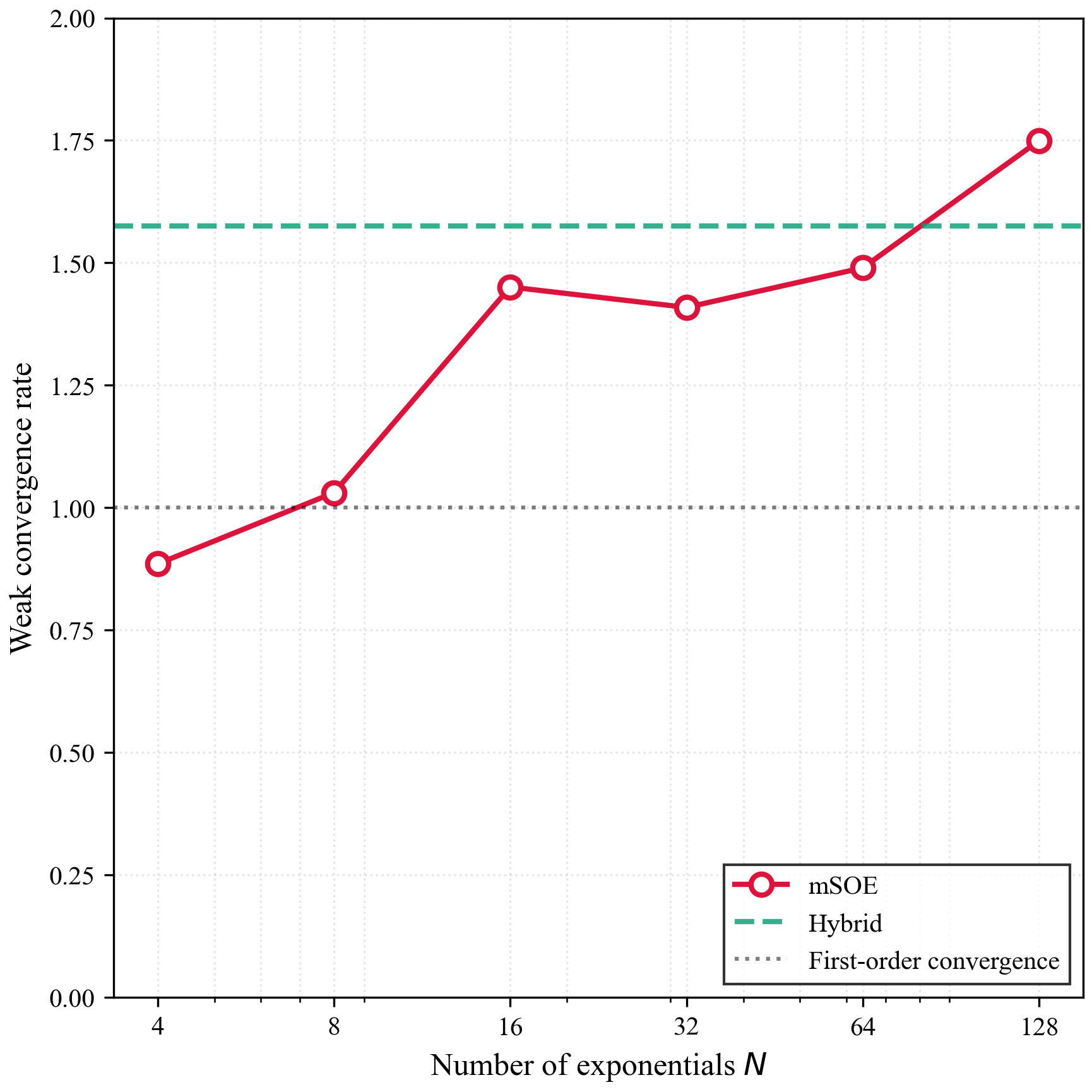}}	
		\caption{Weak convergence rate of maximal relative errors of implied volatility smiles for the mSOE scheme with varying $N$s.}
		\label{fig: weak convergence rate with varying N}
	\end{figure}

	\subsubsection{Implied volatility surface}
	We now examine the performance of schemes in recovering the entire implied volatility surface, utilizing the same set of model parameters as in \eqref{eq: rBergomi param}. The maturities and log-strikes are given as
	\begin{align*}
		T_j = \frac{j}{8},\quad \boldsymbol{k}_j = [-0.1, -0.09, \cdots, 0.05] \times \sqrt{T_j},\quad j = 1, 2,\cdots, 8. 
	\end{align*}
	The results are summarized in Table \ref{tab: maximal rele err for surfaces} and depicted in Figure \ref{fig: max rele error_surface}, which fundamentally align with the findings for the implied volatility smiles. The SOE scheme fails to converge consistently as $n$ increases, whereas both the hybrid benchmark and the mSOE scheme exhibit stable convergence behavior. The hybrid scheme again confirms the importance of treating the near-origin singularity exactly. Compared with the hybrid scheme, the mSOE scheme achieves comparable or better accuracy for sufficiently large $N$. As shown in Figure \ref{fig: weak convergence rate of surface with varying N}, its convergence rate improves systematically as $N$ increases.

	\begin{table}[htp]
		\centering	
		\begin{adjustbox}{max width=\textwidth}
			\begin{tabular}{c|c|cc|cc|cc}
				\toprule	
				& &\multicolumn{2}{c}{$N = 4$} & \multicolumn{2}{c}{$N = 16$} & \multicolumn{2}{c}{$N = 128$} \\
				\hline 
				\makecell{$n$} &Hybrid & SOE & mSOE & SOE & mSOE & SOE & mSOE\\
				\hline 			      
				$128$  &0.0430&0.0123& 0.0249 &0.0172&0.0157 &0.0156& \textbf{0.0150}\\
				$256$  &0.0222&0.0117& 0.0258&0.0196& 0.0079&0.0170& \textbf{0.0080}\\
				$512$  &0.0118&0.0108& 0.0253&0.0202&0.0040&0.0165& \textbf{0.0038}\\
				$1024$  &0.0056&0.0092& 0.0196&0.0195& 0.0012&0.0165& \textbf{0.0019}\\
				$2048$ &0.0027&0.0100& 0.0200&0.0172& 0.0012&0.0161& \textbf{0.0007}\\
				\bottomrule
			\end{tabular}
		\end{adjustbox}
		\caption{Maximal relative error (over maturities and strikes) for the implied volatility surface, with $H = 0.07$ and number of samples $2^{24}$.}
		\label{tab: maximal rele err for surfaces}
	\end{table}
	
	\begin{figure}[hbt!]
		\centering
		\subfigure{
			\includegraphics[width=1\linewidth ,clip]{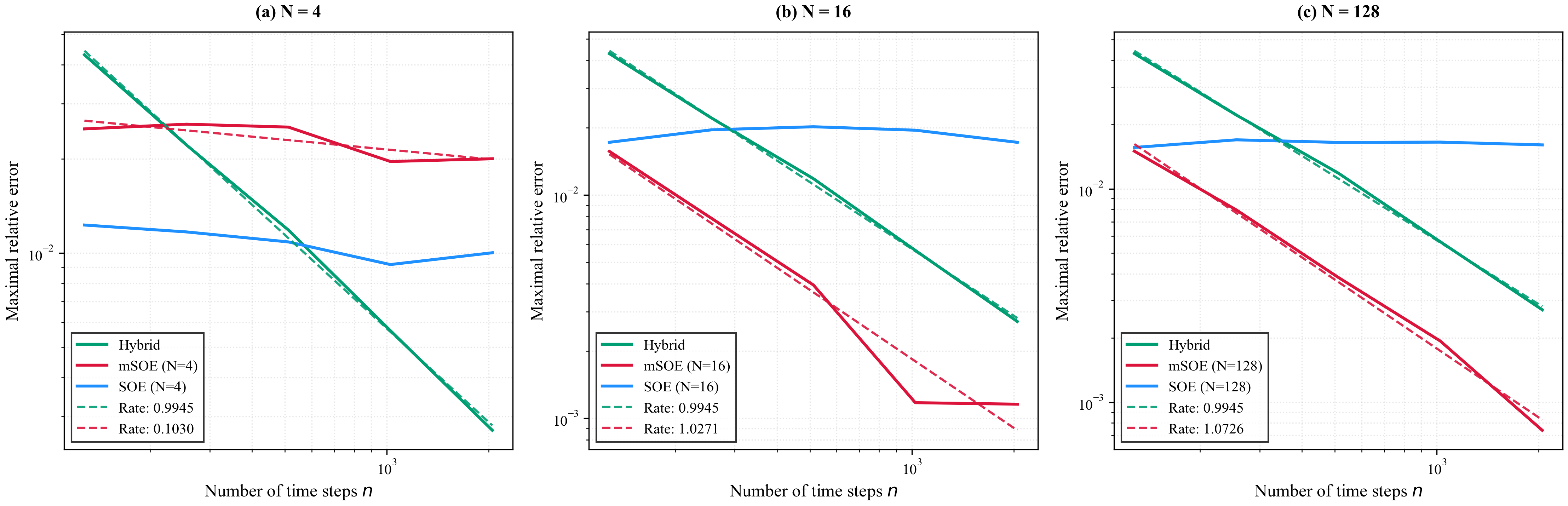}}	
		\caption{Maximal relative errors of implied volatility surfaces using $m = 2^{24}$ samples with $H = 0.07$ and $N = 4$ (left), $N = 16$ (middle) and $N = 128$ (right). The red, blue, and green curves represent the maximal errors between the true surface and surfaces obtained by the mSOE, SOE, and hybrid schemes, respectively. The dashed lines give a reference for the convergence rates.}
		\label{fig: max rele error_surface}
	\end{figure}

	\begin{figure}[hbt!]
		\centering
		\subfigure{
			\includegraphics[width=0.6\linewidth ,clip]{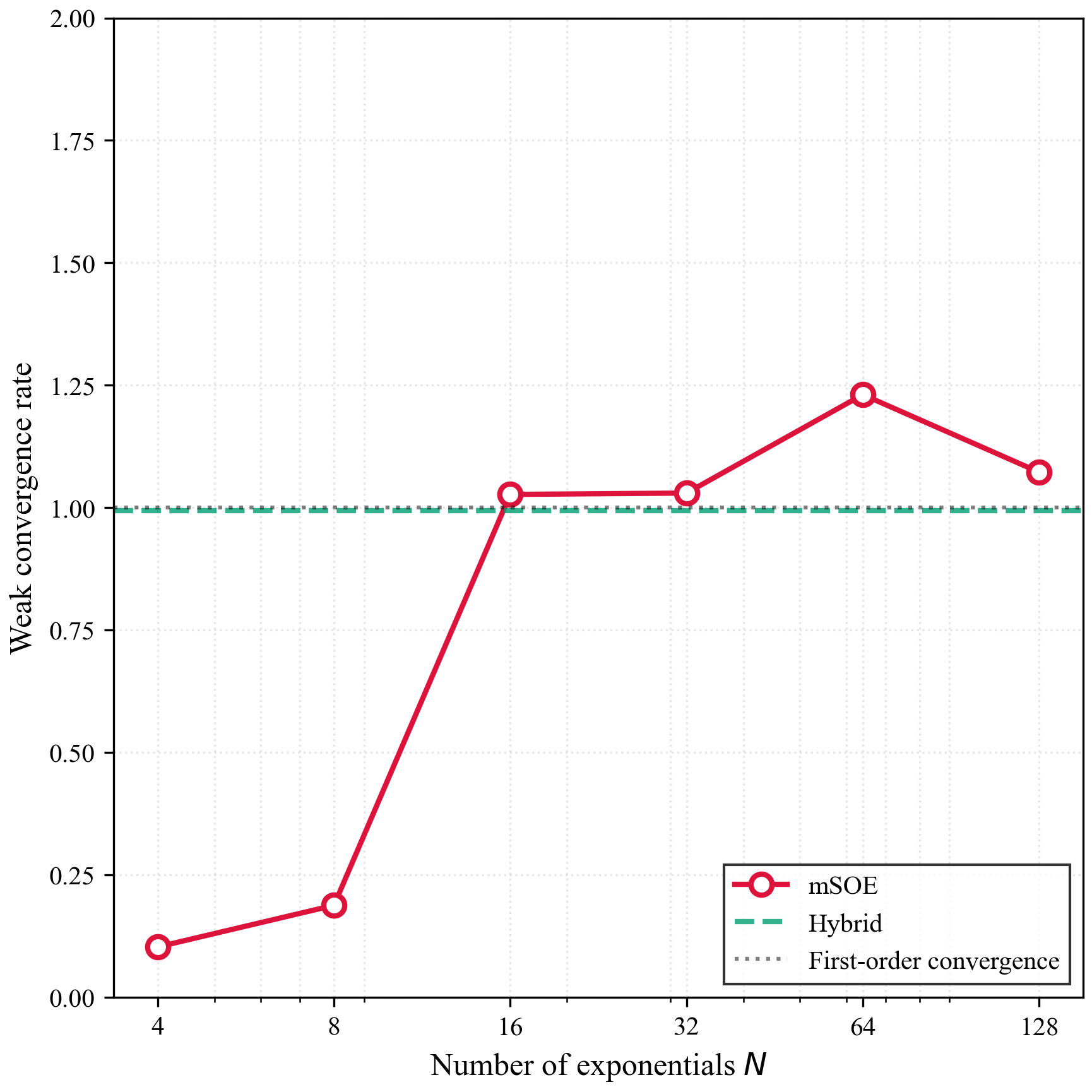}}	
		\caption{Weak convergence rate of maximal relative errors of implied volatility surface for the mSOE scheme with varying $N$s.}
		\label{fig: weak convergence rate of surface with varying N}
	\end{figure}
	
	\subsection{Calibration by Wasserstein-1 distance}\label{subsec: calibration}
	In this subsection, we evaluate the proposed calibration framework in a synthetic-data setting designed to isolate the effect of the calibration loss function. More precisely, for each maturity, the target risk-neutral distribution of the terminal underlying asset price is generated directly from the prescribed ground-truth rBergomi parameters by the Cholesky factorization, and is treated as observable in this subsection. Hence, the experiments reported here should be interpreted as oracle tests. They assess the intrinsic behaviour of the Wasserstein-based calibration objective without the additional uncertainty introduced by recovering the market-implied distribution from finitely many noisy option quotes. 
		
	In real markets, the risk-neutral distribution is not directly observable and must be inferred from option prices (or equivalently, implied volatilities), which is an ill-posed inverse problem and may introduce approximation and regularization errors in tails. Our purpose here is not to bypass this issue, but to examine whether Wasserstein-1 distance improves parameter recovery, optimization stability, and out-of-sample generalization relative to the conventional MSE objective once the target distribution is available. The practical recovery step from option data is discussed separately in Section \ref{subsec: RND extraction}. We emphasize that the goal of this subsection is not to claim a new optimal transport methodology. However, we aim to assess whether the Wasserstein-based distributional matching provides a practical and robust alternative to pointwise price fitting in the setting of rBergomi calibration.

	To isolate the impact of the loss function itself, we adopt a simplified setup where all model parameters $\br{\xi_0, H, \rho, \eta}$ are treated as scalar trainable variables. This allows any observed performance differences to be attributed directly to the choice of loss function. We conduct the calibration process under four distinct parameter scenarios, as listed in Table \ref{tab: model param}. Case 0 serves as the benchmark. Case 1 features a lower volatility level, Case 2 alters the roughness parameter, and Case 3 modifies the correlation and volatility-of-volatility level. The initial guess $\bmtheta_{\text{init}}$ for each case is generated by applying a fixed vector of relative deviation rates to the ground-truth parameters $\bmtheta_{\text{true}}$, ensuring all experiments start from a consistent level of prior knowledge.
	
	The maturities and strikes used for training and testing are specified in Table \ref{tab: contract param}. The training set for each maturity includes five at-the-money and slightly out-of-the-money options, representing the most liquid segment of the market. Deep out-of-the-money options are intentionally excluded from the training set but are included in the test set to evaluate the model's tail-pricing generalization. 
	
	\begin{table}[htp]
		\centering
		\begin{tabular}{c|c|cccc}
			\toprule
			Case & &$\xi_0(t)$ & $H$ & $\rho$ & $\eta$ \\ 
			\hline
			\multirow{2}{*}{0} & $\bmtheta_{\text{true}}$ &0.09 & 0.07 & -0.9 & 1.9\\
			& $\bmtheta_{\text{init}}$ & 0.15 & 0.12 & -0.7 & 1.5\\
			\hline
			\multirow{2}{*}{1} & $\bmtheta_{\text{true}}$ & 0.04 & 0.07 & -0.9 & 1.9\\
			& $\bmtheta_{\text{init}}$ &0.067& 0.12 & -0.7& 1.5\\
			\hline
			\multirow{2}{*}{2} & $\bmtheta_{\text{true}}$& 0.09 & 0.02 & -0.9 & 1.9\\
			& $\bmtheta_{\text{init}}$ &0.15& 0.034& -0.7& 1.5\\
			\hline
			\multirow{2}{*}{3} & $\bmtheta_{\text{true}}$ & 0.09 & 0.07 & -0.7 & 2.2 \\
			& $\bmtheta_{\text{init}}$&0.15& 0.12& -0.544& 1.737\\
			\bottomrule
		\end{tabular}
		\caption{The set of ground-truth $\bmtheta_{\text{true}}$ and the initial guess $\bmtheta_{\text{init}}$ of model parameters.}
		\label{tab: model param}
	\end{table}
	
	\begin{table}[htp]
		\centering
		\begin{tabular}{lll}
			\toprule
			& Train set & Test set \\
			\hline
			$T$& \multicolumn{2}{l}{$[0.3, 0.5, 1.0]$} \\		
	        $K$ &  $[0.9, 0.95, 1.0, 1.05, 1.1]$  & $[0.8, 0.85, 1.15, 1.2]$ \\			
			\bottomrule
		\end{tabular}
		\caption{The set of maturities and strikes used for training and testing.}
		\label{tab: contract param}
	\end{table}

	\subsubsection{Calibration with European Options}\label{subsubsec: calibration_to_European_options}
	We first compare the performance of Wasserstein-1 and MSE objectives in the calibration of European put and call options. The benchmark price surface $\{P_j(\bmtheta_{\text{true}})\}_{j=1}^M$ and the target terminal samples $\{S_{T_j}(\boldsymbol{\theta}_{\text{true}})\}_{j=1}^M$ are simulated with time discretization size $\tau = 1/500$ and Monte Carlo repetitions $m = 2^{15}$ using the Cholesky factorization. During the calibration process, we use the same setting to generate samples of $\{S_{T_j}(\boldsymbol{\theta})\}_{j=1}^M$ for each candidate parameter vector $\bmtheta$ and compute $\{P_j(\bmtheta)\}_{j=1}^M$ by the mSOE scheme with the kernel approximation tolerance $\varepsilon = 10^{-5}$. The number of exponential terms $N$ is adapted dynamically according to the current $H$ so as to maintain this prescribed approximation precision $\varepsilon$. Both Wasserstein and MSE calibration problems are solved by the box-constrained L-BFGS-B algorithm \cite{byrd1995limited, zhu1997algorithm}, while automatic differentiation is implemented in TensorFlow \cite{abadi2016tensorflow} for automatic differentiation. The model parameters are subject to the following constraints:
	\begin{align*}
	\xi_0 \in [0.001, 0.3],\quad H \in [0.01, 0.499],\quad \rho \in [-0.999, -0.1],\quad \eta \in [1, 4].	
	\end{align*}
	We set the stopping tolerance to be $\epsilon_f^\text{wass} = 10^{-10}, \epsilon_g^{\text{wass}} = 10^{-6}, \epsilon_f^{\text{mse}} = 10^{-12}, \epsilon_g^{\text{mse}} = 10^{-8}$. The maximum number of iterations is set to be 500. 	
		
	The calibration results are summarized in Table \ref{tab: calibration_performance}. Here, “MSE-Narrow” denotes direct price fitting on the narrow strike range listed in Table \ref{tab: contract param}, whereas “W1-Direct” denotes Wasserstein calibration using the oracle target distributions. To assess computational efficiency, the third and fourth columns report the number of completed iterations and the total calibration time, respectively. We evaluate the in-sample and out-of-sample pricing accuracy by the root mean squared error (RMSE)
	\begin{align*}
		\text{RMSE} &= \sqrt{\frac{1}{M}\sum_{j = 1}^{M}(P_j(\bmtheta^{\ast}) - P_j(\bmtheta_{\text{true}}))^2},
	\end{align*}
	and the maximum absolute percentage error (MaxAPE)
	\begin{align*}
		\text{MaxAPE} = \max_{j =1, \cdots, M}\hua{\abs{\frac{P_j(\bmtheta^{\ast}) - P_j(\bmtheta_{\text{true}})}{P_j(\bmtheta_{\text{true}})}}},
	\end{align*}
	where $\bmtheta^{\ast}$ denotes the calibrated parameter vector. The RMSE measures the overall pricing accuracy, while MaxAPE records the worst-case relative error. The in-sample pricing errors are presented in columns 5-6, while the out-of-sample errors are shown in columns 7-8.

	As evidenced in Table \ref{tab: calibration_performance}, the calibration based on the Wasserstein-1 distance consistently yields more accurate results than the MSE-based calibration. It leads to uniformly smaller out-of-sample RMSE and MaxAPE across all parameter cases. The in-sample pricing errors of these two methods are of comparable magnitude. This is consistent with the fact that the MSE objective is defined only on a discrete set of option prices, whereas the Wasserstein distance matches the terminal distributions of each selected maturity and therefore captures more tail information of the model. From the computational perspective, the Wasserstein-based method does not uniformly reduce the total calibration time, but it generally attains competitive accuracy with a moderate number of iterations. 

	\begin{table}[htp]
		\centering	
		\begin{adjustbox}{max width=\textwidth}
			\begin{tabular}{c|c|cc|cc|cc}
				\toprule	
				\multirow{2}{*}{}&\multirow{2}{*}{}&\multicolumn{2}{c}{Training} & \multicolumn{2}{c}{in-sample} & \multicolumn{2}{c}{out-of-sample} \\
				\hline 
				Case & Method&\# Iters & Total time (s) & RMSE &MaxAPE &RMSE & MaxAPE\\
				\hline 
				\multirow{2}{*}{0} & MSE-Narrow &38 & 317.40& 0.0002& 0.0122& 0.0007& 0.4077\\
				& W1-Direct &67 & 553.59& 0.0003& 0.0074& 0.0002& 0.0641\\	
				\hline 
				\multirow{2}{*}{1} & MSE-Narrow &51 & 488.31& 0.0002& 0.0718& 0.0003&0.4811\\
				& W1-Direct & 125 & 1028.03& 0.0003 & 0.0271& 0.0001& 0.1180\\	
				\hline 
				\multirow{2}{*}{2} & MSE-Narrow & 14& 141.40& 0.0003 & 0.0195& 0.0003&0.0913\\
				& W1-Direct & 89& 734.71 & 0.0003& 0.0182& 0.0003& 0.0813\\	
				\hline 
				\multirow{2}{*}{3} & MSE-Narrow & 26 & 255.99& 0.0003& 0.0200& 0.0006&0.1482\\
				& W1-Direct & 41& 344.90&  0.0005& 0.0186& 0.0003& 0.0596\\	
			   \bottomrule
			\end{tabular}		
		\end{adjustbox}
		\caption{The calibration performance comparison between MSE- and Wasserstein-based calibrations.}
		\label{tab: calibration_performance}
	\end{table}
	
	To further assess the parameter recovery property, we present in Table \ref{tab: parameter_recovery_ape} the absolute percentage errors (APE) between the calibrated  parameters and ground-truth values, defined by
	\begin{align*}
		\text{APE} = \abs{\frac{\theta^\ast - \theta_{\text{true}}}{\theta_{\text{true}}}}.
	\end{align*}
The Wasserstein-based calibration exhibits consistently smaller parameter errors for all four parameters in every parameter case. The improvement is especially pronounced for Cases 0 and 3, for which the MSE-based calibration produces substantial deviations from the true values even when the in-sample pricing errors remain small. This observation highlights a clear identifiability advantage of the Wasserstein objective. By calibrating directly to terminal asset price distributions rather than to a discrete collection of option prices, Wasserstein distance could effectively recover the data-generating parameter set and reduce the risk of overfitting to the training contracts. We also note that the recovery of $H$ remains challenging in the extremely rough regime, especially in Case 2. We therefore interpret the results as showing improved identifiability relative to MSE, rather than uniformly precise recovery across all parameter regimes.
	
	\begin{table}[htp]
		\centering
		\begin{adjustbox}{max width=\textwidth}
			\begin{tabular}{c|c|cccc}
				\toprule
				Case &Method& $\xi_0(t)$ & $H$ & $\rho$ & $\eta$ \\
				\hline
				\multirow{2}{*}{0} & MSE-Narrow& 0.0416 & 0.2953& 0.0328& 0.1822\\
				& W1-Direct& 0.0015 & 0.0407& 0.0077& 0.0114\\ 
				\hline
				\multirow{2}{*}{1} & MSE-Narrow& 0.0134& 0.0906& 0.0069& 0.0606\\
				& W1-Direct& 0.0010& 0.1073 & 0.0397& 0.0380\\ 
				\hline
				\multirow{2}{*}{2} & MSE-Narrow& 0.0042& 1.3787& 0.2212 & 0.2099\\
				& W1-Direct& 0.0015& 1.0841& 0.2107& 0.1834\\ 
				\hline
				\multirow{2}{*}{3} & MSE-Narrow&0.0580& 0.7036& 0.1215& 0.2087\\
				& W1-Direct& 0.0046& 0.0372& 0.0182& 0.0178\\ 
				\bottomrule
			\end{tabular}
		\end{adjustbox}
		\caption{Absolute percentage error between $\bmtheta_{\text{true}}$ and $\bmtheta^\ast$.} 
		\label{tab: parameter_recovery_ape}
	\end{table}
	To directly evaluate the calibration quality in implied volatilities, we present in Figure \ref{fig: iv_surface_comparison_case0} the target and calibrated implied volatility (IV) surfaces together with corresponding signed error maps for Case 0, which is representative of the behaviour observed across the four scenarios in Table \ref{tab: calibration_performance}. The upper panels compare the reference IV surface with the calibrated surfaces obtained under the Wasserstein and MSE objectives, respectively. Both calibration methods recover the general shape of the target surface, but the Wasserstein-based approach appears to provide a closer approximation to the reference surface in a global sense. This observation is supported by the lower panels, which display the signed IV errors pointwise over the maturity-strike grid. The Wasserstein calibration leads to smaller and more spatially uniform errors, whereas the MSE-based calibration shows more significant local discrepancies. Furthermore, while both methods fit the training contracts reasonably well, the Wasserstein objective appears to deliver better extrapolation performance outside the training set. These results indicate that the Wasserstein formulation offers a more robust characterization of the surface structure in Case 0.	
	
	\begin{figure}[hbt!]
		\centering
		\subfigure{
			\includegraphics[width=0.99\linewidth ,clip]{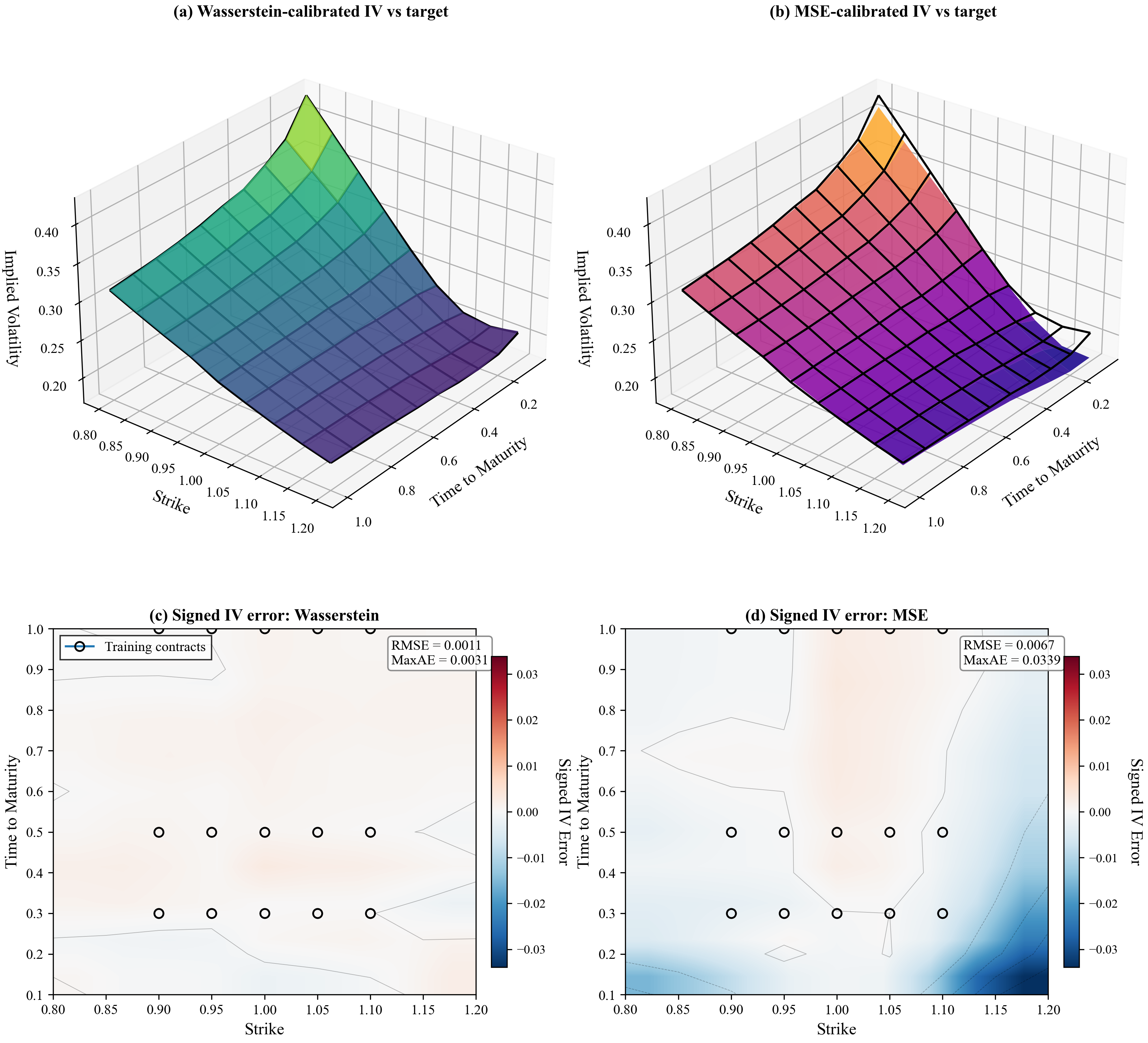}}	
		\caption{Target and calibrated implied volatility surface with signed error maps for Case 0.  Upper panel: Wasserstein- and MSE-calibrated surfaces. The black wireframe represents the target implied volatility surface. Lower panel: the associated pointwise signed IV errors. Circles denote the training contracts. The RMSE and maximum absolute error (MaxAE) are computed based on the calibrated IV surfaces.}
		\label{fig: iv_surface_comparison_case0}
	\end{figure}

	\subsubsection{Generalization to Barrier Options}
	To further evaluate the calibrated model's capacity to capture the tail behavior of the underlying asset's distribution, we consider out-of-sample pricing experiments involving barrier options. Since the payoff depends on the entire price path through a barrier crossing event, these contracts are more sensitive to tail behaviours and path dynamics than European options used in calibration. We consider two complementary products:
	
	\begin{enumerate}
		\item Down-and-out put (DOP), which is sensitive to downside movements and probes the left tail. A DOP has the same terminal payoff as a European put, provided that the underlying price does not fall below a lower barrier level $B$ during $[0, T]$. Its payoff is given by
		\begin{equation*}
			h_{\text{DOP}}(S_T) = \left\{
			\begin{aligned}
				&\max(K - S_T, 0), &\text{if}\ \min_{t\in [0,T]} S_t > B, \\
				&0, &\text{otherwise}.
			\end{aligned}
			\right.
		\end{equation*}
		\item Up-and-out call (UOC), which is sensitive to upward movements and tests the right tail. A UOC has the same terminal payoff as a European call, provided that the underlying price does not rise above an upper barrier level $B$ during $[0, T]$. Its payoff is given by 
		\begin{equation*}
			h_{\text{UOC}}(S_T) = \left\{
			\begin{aligned}
				&\max(S_T - K, 0), &\text{if}\ \max_{t\in [0,T]} S_t < B, \\
				&0, &\text{otherwise}.
			\end{aligned}
			\right.
		\end{equation*}
	\end{enumerate}
	\begin{table}[htp]
		\centering
		\begin{tabular}{c|c|c}
			\toprule
			& DOP & UOC\\
			\hline
			$T$ & \multicolumn{2}{c}{$[0.3, 0.5, 1.0]$} \\
			\hline
			$K$ & 0.95 & 1.05 \\
			\hline
			$B$ & $[0.70, 0.71, \cdots, 0.85]$ & $[1.15, 1.16, \cdots, 1.30]$ \\				
			\bottomrule
		\end{tabular}
		\caption{Contract parameters for the barrier-option experiments}
		\label{tab: contract param of barrier options}
	\end{table}
	The contract parameters for the DOP and UOC tests are reported in Table \ref{tab: contract param of barrier options}. To rigorously assess the model's generalization capability, the barrier levels $B$ are deliberately set outside the range of strikes used for calibration in Table \ref{tab: contract param}. In addition, to address statistical uncertainty in the pricing results, we estimate for each barrier contract the benchmark price by Monte Carlo simulation together with a 95\% confidence interval, using Cholesky factorization.

	Figure \ref{fig: barrier_options_analysis_set_0_uoc} and \ref{fig: barrier_options_analysis_set_0_dop} present the prices of UOC and DOP options computed over a range of barrier levels using the model parameters calibrated in Case 0, respectively. In each figure, the black curve represents the benchmark price obtained from the ground-truth parameters, the shaded band shows its 95\% confidence interval.
	For UOC, the pricing curves generated by Wasserstein-calibrated parameters (red) are nearly indistinguishable from the benchmark across all maturities and barrier levels, and in most cases remain within the benchmark confidence bands. By contrast, the MSE-based calibration (blue) exhibits systematic upward deviations, especially for shorter maturities, and many prices fall outside the 95\% confidence interval. For DOP contracts, the difference between the two methods is smaller, but the Wasserstein-based calibration still tracks the benchmark more accurately and consistently remains within the confidence bands. 
	
Moreover, Table \ref{tab: barrier_options_analysis_summary} summarizes the barrier option pricing errors for the two calibration methods across four parameter sets. We report the mean absolute error (MAE), root mean squared error (RMSE), mean absolute percentage error (MAPE), and the fraction of barrier contracts for which the price implied by the calibrated model lies outside the 95\% confidence interval of the corresponding benchmark (Significant 95\%). Overall, Wasserstein-based calibration performs better than MSE in most cases. The improvement is particularly pronounced for UOC contracts, where W1-Direct consistently delivers smaller errors in Cases 0, 1, and 3, together with a much lower proportion of contracts exhibiting statistically significant deviations from the true benchmark. For DOP contracts, the difference between the two methods is generally smaller, but Wasserstein-based calibration remains more robust and typically achieves smaller errors. Even though Case 2 is comparatively mixed, this table indicates that Wasserstein-based calibration better preserves the tail information relevant for barrier options. 
		\begin{figure}[hbt!]
		\centering
		\subfigure{
			\includegraphics[width=0.99\linewidth ,clip]{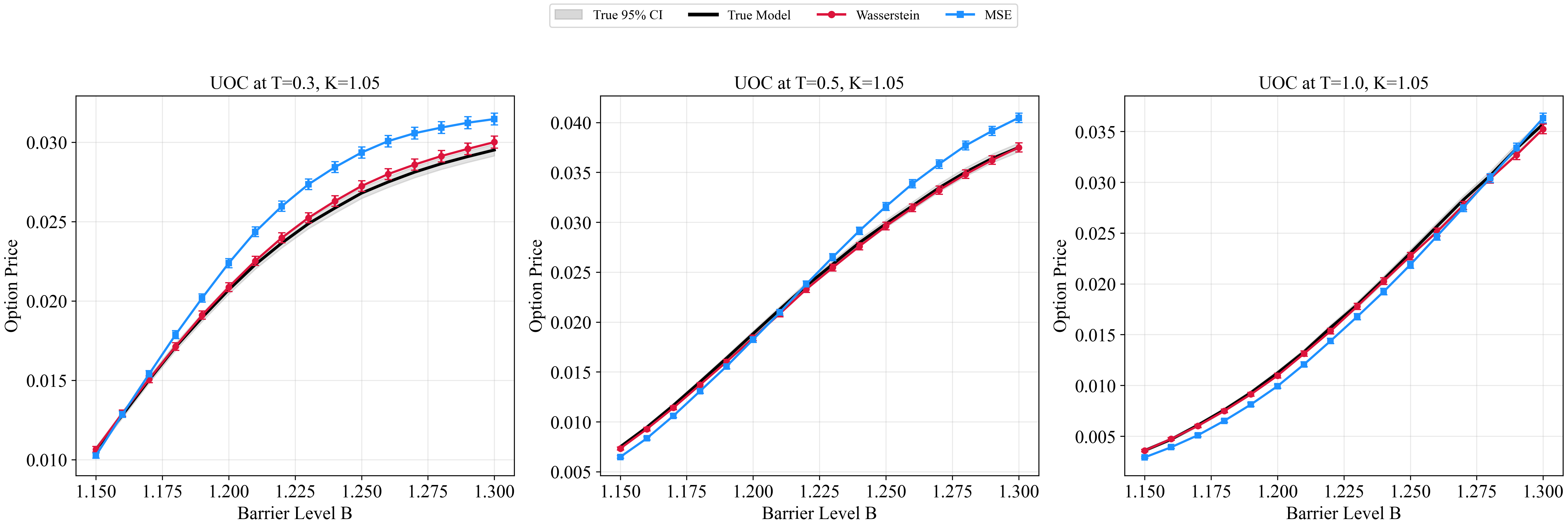}}	
		\caption{Comparison of UOC prices across barrier levels between Wasserstein and MSE calibration. }
		\label{fig: barrier_options_analysis_set_0_uoc}
	\end{figure}

	\begin{figure}[hbt!]
		\centering
		\subfigure{
			\includegraphics[width=0.99\linewidth ,clip]{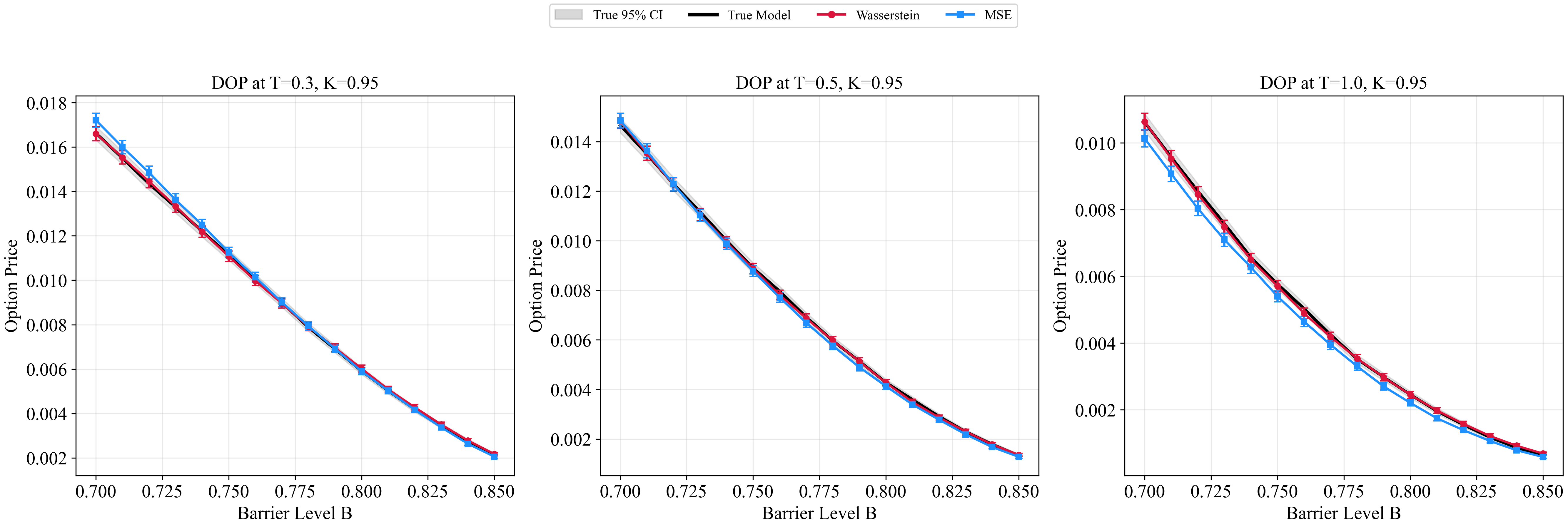}}	
		\caption{Comparison of DOP prices across barrier levels between Wasserstein and MSE calibration. }
		\label{fig: barrier_options_analysis_set_0_dop}
	\end{figure}	
	
	\begin{table}[htp]
		\centering
		\begin{tabular}{c|c|c|c ccc}
			\toprule
			Set & Option & Method & MAE & RMSE & MAPE(\%) & Significant 95\% \\
			\hline
			\multirow{4}{*}{0} 
			& \multirow{2}{*}{UOC} & W1-Direct & 0.0003 & 0.0003& 1.3139& 0.1042\\ \cline{3-7}
			&                      & MSE-Narrow  &0.0014& 0.0016& 7.2243& 0.8750\\ \cline{2-7}
			& \multirow{2}{*}{DOP} & W1-Direct &0.0001& 0.0001& 1.1914& 0\\ \cline{3-7}
			&                      & MSE-Narrow  & 0.0002& 0.0003& 4.1823& 0.4375\\ \hline
			
			\multirow{4}{*}{1} 
			& \multirow{2}{*}{UOC} & W1-Direct & 0.0001& 0.0001& 0.5605& 0\\ \cline{3-7}
			&                      & MSE-Narrow & 0.0004& 0.0005& 2.3202& 0.5417\\ \cline{2-7}
			& \multirow{2}{*}{DOP} & W1-Direct & 0.0001& 0.0001 & 0.9376 & 0\\ \cline{3-7}
			&                      & MSE-Narrow & 0.0001& 0.0001& 1.7750& 0.0833\\ \hline
			
			\multirow{4}{*}{2} 
			& \multirow{2}{*}{UOC} & W1-Direct & 0.0003 & 0.0003& 2.4619& 0.4167\\ \cline{3-7}
			&                      & MSE-Narrow & 0.0003 & 0.0003& 1.6498& 0.3333\\ \cline{2-7}
			& \multirow{2}{*}{DOP} & W1-Direct & 0.0003 & 0.0003& 2.7565& 0.3542\\ \cline{3-7}
			&                      & MSE-Narrow & 0.0004& 0.0005& 4.5288& 0.8125\\ \hline
			
			\multirow{4}{*}{3} 
			& \multirow{2}{*}{UOC} & W1-Direct & 0.0001& 0.0002& 0.8829& 0\\ \cline{3-7}
			&                      & MSE-Narrow  &0.0008& 0.0009& 6.0428& 0.8333\\ \cline{2-7}
			& \multirow{2}{*}{DOP} & W1-Direct & 0.0001& 0.0002& 1.5392& 0.0625\\ \cline{3-7}
			&                      & MSE-Narrow  & 0.0005 & 0.0006& 4.9409& 0.6667\\
			\bottomrule
		\end{tabular}
		\caption{Summary statistics for barrier-option pricing errors and 95\% significance rates.}
		\label{tab: barrier_options_analysis_summary}
	\end{table}

	\subsubsection{Loss landscape}\label{subsubsec: loss landscape}
	To compare the optimization landscapes induced by two calibration losses, Figures \ref{fig: 2D_contour_plot_wass} and \ref{fig: 2D_contour_plot_mse} present two-dimensional contour plots of the Wasserstein-1 and MSE losses for all parameter pairs in $(\xi_0, H, \rho, \eta)$. In each panel, the selected parameter pair is varied over a $25 \times 25$ grid, while the remaining parameters are fixed at their ground-truth values. For visual clarity, we plot the logarithm of the loss. The red star marks the true parameter value in Case 0 of Table \ref{tab: model param}, and the yellow circle denotes the empirical minimizer on the numerical grid. The losses are evaluated under the same numerical setting as in the calibration experiments.
	
	The contour plots empirically suggest different optimization behaviors for the two losses. The MSE contours show narrower valleys and stronger distortion with more pronounced curvature variation, especially for the pairs $(\xi_0, \eta), (\rho, H), (\eta, H)$. It indicates a more rugged and less stable optimization problem. By contrast, the Wasserstein-1 contours are visually smoother and closer to elliptical, with more clearly localized basins around the minimizer in these experiments. Moreover, the empirical minimizers under Wasserstein-1 tend to lie closer to the true parameters. The orientation and narrowness of several MSE contours also suggest stronger parameter coupling, while the Wasserstein-1 contours indicate weaker coupling. This helps to explain the improved parameter identifiability observed in the calibration experiments.

	\begin{figure}[hbtp]
		\centering
		\subfigure{
			\includegraphics[width=1\linewidth,clip]{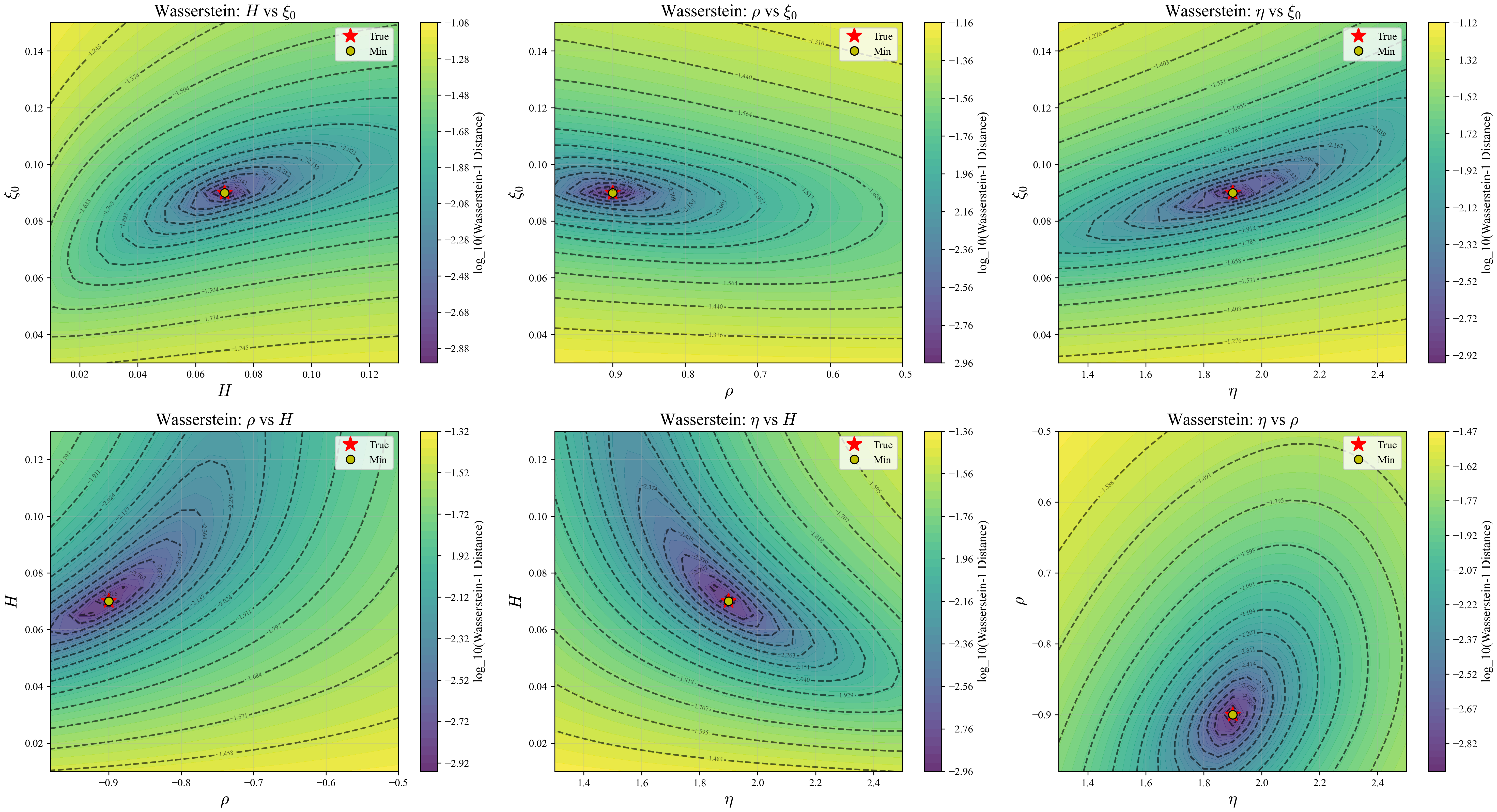}}	
		
		\caption{Two-dimensional contour plots of the logarithm of the Wasserstein-1 calibration loss for all parameter pairs in $(\xi_0,H,\rho,\eta)$.}
		\label{fig: 2D_contour_plot_wass}
	\end{figure}

	\begin{figure}[hbtp]
		\centering
		\subfigure{
			\includegraphics[width=1\linewidth,clip]{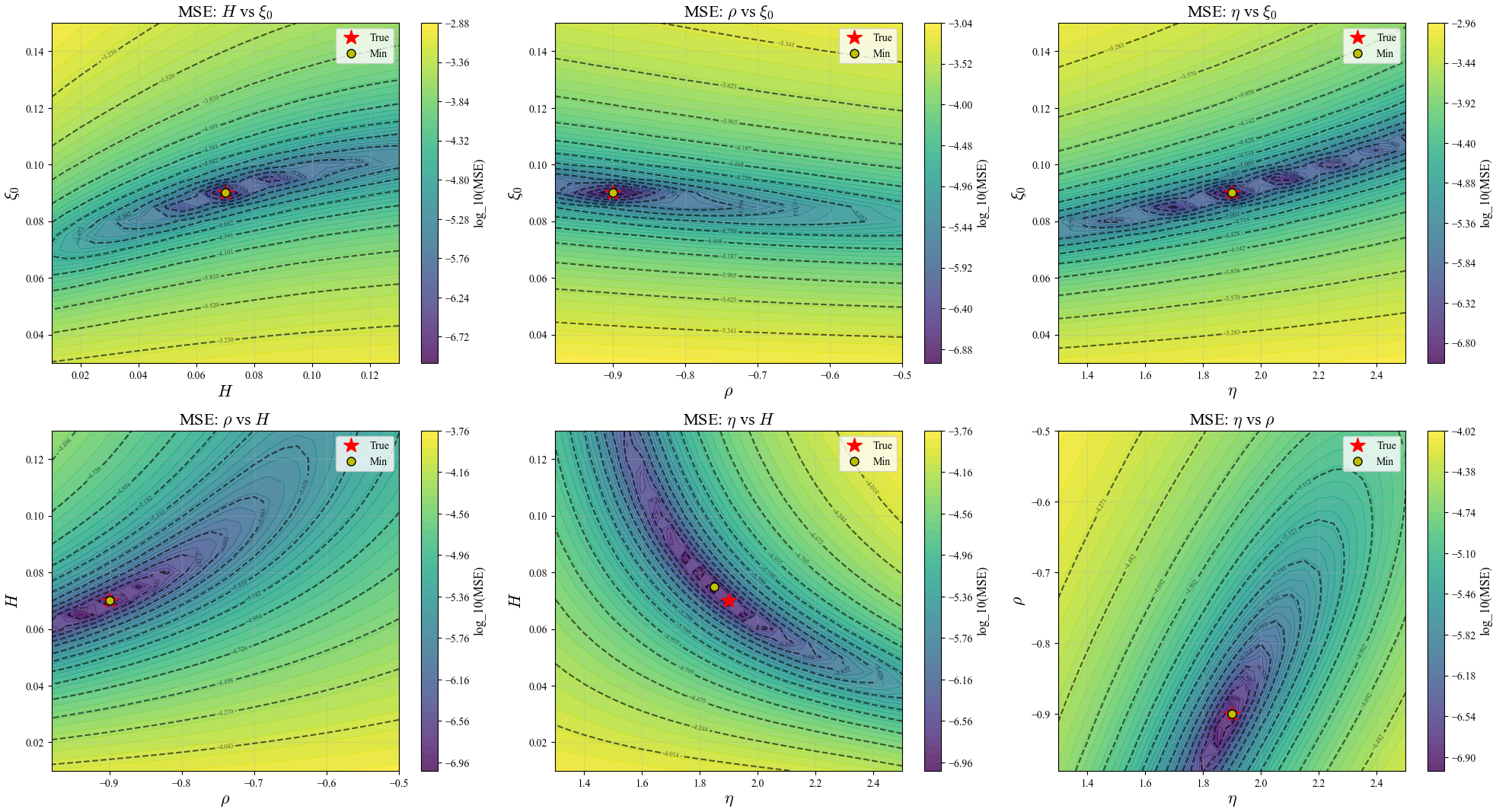}}
		\caption{Two-dimensional contour plots of the logarithm of the MSE calibration loss for all parameter pairs in $(\xi_0, H,\rho,\eta)$.}
		\label{fig: 2D_contour_plot_mse}
	\end{figure}

	\subsection{Recovering the target distribution via SVI}\label{subsec: RND extraction}
	Section \ref{subsec: calibration} considers a setting in which the target risk-neutral distribution is directly generated from the ground-truth model parameters. In practice, such a distribution is not readily observable and must be recovered from discrete option prices. We therefore complement the oracle experiments by introducing a distribution-recovery step based on the SVI (stochastic volatility inspired) parameterization \cite{gatheral2004parsimonious}. The purpose of this subsection is twofold: first, to bridge the gap between the idealized setup of Section \ref{subsec: calibration} and a more realistic calibration workflow; second, to quantify the approximation error introduced when the target distribution is reconstructed from a finite set of European option prices. 
	
	Specifically, we start from the synthetic option prices generated by Cholesky factorization on the maturity set $T \in \hua{0.3, 0.5, 1.0}$, together with the narrow range of strike grid $K$ listed in Table \ref{tab: contract param} and the wide range of strike grid in Table \ref{tab: wide quotes}. From these prices we first compute Black-Scholes implied volatilities $\sigma_{BS}(K, T)$ and then the corresponding total variances $w(k, T) := \sigma_{BS}^2(k, T)T$, where $k = \log(K/F_t)$ is log-moneyness and $F_t = \bbE[S_t | \calF_0]$ is the forward price. The resulting total-variance slices are fitted independently across maturities using the raw SVI parameterization \cite{gatheral2014arbitrage}.  
	
	For each fixed maturity $T$, we represent the total implied variance as 
	\begin{align}\label{eq: SVI}
		w(k) = a + b\left\{ \rho(k-m) + \sqrt{(k-m)^2 + \sigma^2}\right\}, 
	\end{align}
	where the parameter set $(a, b, \rho, m, \sigma)$ satisfies $a \in \R, b \geq 0, \abs{\rho} < 1, m \in \R, \sigma > 0$ and the condition $a + b\sigma\sqrt{1 - \rho^2} \geq 0$, which ensures the non-negativity of $w(k)$ for all $k \in \R$. Each parameter explicitly controls the shape of volatility smiles, and we refer readers to \cite{gatheral2014arbitrage} for their specific effects. In the implementation, these parameters are calibrated by a two-step routine. It first optimizes $(m, \sigma)$ by minimizing the least-squares fitting error of the total variance slice, and then solves a constrained linear least-squares problem for $(a, b, \rho)$. This procedure is repeated until the improvement in mean squared error falls below a prescribed tolerance.  
    
    The fitted SVI slices provide a smooth proxy of the implied volatility smiles from discrete option observations. Based on this parameterization, we recover the maturity-wise risk-neutral distribution through the analytical cumulative distribution function (cdf) by the Breeden-Litzenberger equation. In particular, let $w^\prime(k)$ denote the derivative of the total variance with respect to log-moneyness, the risk-neutral cdf has the representation:
    \begin{align*}
    	F_{S_T}(K) = \Phi(-d_2) + \phi(d_2)\frac{w^\prime(k, T)}{2\sqrt{w(k, T)}},\quad d_2 = \frac{-k}{\sqrt{w(k)}} - \frac{\sqrt{w(k)}}{2}, 
    \end{align*}
    where $\Phi(\cdot)$ and $\phi(\cdot)$ are the standard normal cdf and density, respectively. To construct the target samples required by the Wasserstein calibration algorithm, we do not simulate directly from the fitted SVI model. Alternatively, we generate a set of equally spaced probability levels and compute the corresponding quantiles by numerically inverting the recovered cdf:
    \begin{align*}
    	S_{T, (i)}^{MKT} = F_{S_T}^{-1}\br{\frac{i-0.5}{m}},\quad i = 1, \cdots, m. 
    \end{align*}
    This quantile-based construction is particularly convenient for Wasserstein calibration, since the empirical Wasserstein-1 distance in one dimension can be computed directly from sorted samples or quantile functions. Therefore, once the SVI-implied quantiles have been obtained, they can be used in exactly the same way as the oracle quantiles in Section \ref{subsec: calibration}. The only difference is the non-negligible recovery error inherited from the discrete option data in the SVI fitting procedure. 
\begin{table}[htp]
	\centering
	\begin{tabular}{lll}
		\toprule
		& train set & test set \\
		\hline
		$T$& \multicolumn{2}{l}{$[0.3, 0.5, 1.0]$} \\		
		$K$ &  $[0.8, 0.85, 0.9, 0.95, 1.0, 1.05, 1.1, 1.15, 1.2]$  & $[0.7, 0.75, 1.25, 1.3]$ \\			
		\bottomrule
	\end{tabular}
	\caption{The set of maturities and strikes used for training and testing.}
	\label{tab: wide quotes}
\end{table}

	\begin{table}[htp]
		\centering
		\begin{tabular}{c|c|c|c|c}
			\toprule
			Set& $T$ & W1-benchmark & W1-Narrow & W1-Wide \\
			\hline
			\multirow{3}{*}{0}& 0.3 & 0.0011& 0.0017& 0.0012\\
			& 0.5 & 0.0013& 0.0037&  0.0020\\
			& 1 & 0.0014 & 0.0135& 0.0051\\		
			\hline
			\multirow{3}{*}{1}& 0.3 & 0.0008& 0.0011& 0.0009\\
			& 0.5 & 0.0008& 0.0014& 0.0011\\
			& 1 & 0.0010& 0.0046&  0.0027\\		
			\hline
			\multirow{3}{*}{2}& 0.3 & 0.0011 & 0.0007& 0.0008\\
			& 0.5 & 0.0020& 0.0042& 0.0017\\
			& 1 & 0.0026& 0.0525& 0.0054\\		
			\hline
			\multirow{3}{*}{3}& 0.3 & 0.0016 & 0.0053& 0.0016\\
			& 0.5 & 0.0019& 0.0106&  0.0017\\
			& 1 & 0.0025& 0.0061&  0.0045\\				
			\bottomrule
		\end{tabular}
		\caption{Empirical Wasserstein distance comparison. }
		\label{tab: recovery property}
	\end{table}

	To evaluate the recovery error, Table \ref{tab: recovery property} compares three empirical Wasserstein-1 distances. The first column, denoted by W1-benchmark, is the empirical Wasserstein-1 distance between two independent Monte Carlo samples from the true model, which serves as a sampling-error baseline. W1-Narrow and W1-Wide are the Wasserstein-1 distances between the true distribution and the SVI-recovered distribution using the narrow and wide strike set, respectively. The results show that the width of the strike interval has a substantial impact on the recovery accuracy. In almost all cases, W1-Wide is smaller than W1-Narrow. The improvement is particularly pronounced for longer maturities. This confirms that including a wider range of strikes improves the identification of the implied underlying asset price distribution. Meanwhile, the table also indicates that the SVI-based recovery remains imperfect even under the wide strike range. W1-Wide is consistently above the W1-benchmark, especially for the longer maturity scenarios, which implies that the final Wasserstein-based calibration error must be attributed to this intermediate distribution-recovery step. Nevertheless, for some short maturity cases, the wide-range recovery is already very close to the Monte Carlo baseline. In these cases, the reconstructed distribution is sufficiently accurate for practical purposes. 
	
		\begin{table}[htp]
		\centering	
		\begin{adjustbox}{max width=\textwidth}
			\begin{tabular}{c|l|cc|cc|cc}
				\toprule	
				\multirow{2}{*}{}&\multirow{2}{*}{}&\multicolumn{2}{c}{Training} & \multicolumn{2}{c}{in-sample} & \multicolumn{2}{c}{out-of-sample} \\
				\hline 
				Set & Method &\# Iters & Total time (s) & RMSE &MaxAPE &RMSE & MaxAPE\\
				\hline 
				\multirow{3}{*}{0} & MSE-Wide &49 & 467.08& 0.0002& 0.0953& 0.0002& 0.3100\\
				& W1-Recovered-Wide & 71& 677.05& 0.0004& 0.0428& 0.0004& 0.1395\\	
				& W1-Recovered-Narrow &62& 592.42& 0.0006& 0.0365& 0.0009& 0.2686\\
				\hline 
				\multirow{3}{*}{1} & MSE-Wide &53 & 511.91& 0.0002& 0.2444& 9.0581e-5& 0.1636\\
				& W1-Recovered-Wide &74& 698.69& 0.0003& 0.0791& 0.0002& 1.0816\\	
				& W1-Recovered-Narrow &64 & 609.12& 0.0003& 0.0374&0.0003& 0.1544\\
				\hline 
				\multirow{3}{*}{2} & MSE-Wide & 12 & 122.72& 0.0003 & 0.0872& 0.0003 & 0.2693\\
				& W1-Recovered-Wide & 116 & 1089.23& 0.0004& 0.0612& 0.0004& 0.1160\\	
				& W1-Recovered-Narrow & 50 & 485.83& 0.0015& 0.0446& 0.0012&  0.0731\\
				\hline  
				\multirow{3}{*}{3} & MSE-Wide &41 & 393.35& 0.0003& 0.0394& 0.0002& 0.1824\\
				& W1-Recovered-Wide &72 &682.45& 0.0006 & 0.0373& 0.0004& 0.0578\\	
				& W1-Recovered-Narrow &60 & 572.26& 0.0006& 0.0213& 0.0005& 0.1499\\				
				\bottomrule
			\end{tabular}		
		\end{adjustbox}
		\caption{The calibration performance comparison between MSE and Wasserstein-1 distance as loss functions.}
		\label{tab: calibration-performance-recovered}
	\end{table}
	
	To study how the recovery error propagates to the calibration results, Table \ref{tab: calibration-performance-recovered} compares three calibration strategies: direct price fitting by MSE on the wide strike range (MSE-Wide), Wasserstein-based calibration using the SVI-recovered target distribution from the wide strike range (W1-Recovered-Wide), and the Wasserstein calibration based on the narrow-strike SVI recovery (W1-Recovered-Narrow). This table indicates that W1-Recovered-Wide gives the most robust calibration performance across the four parameter sets. It tends to achieve that best balance between the in-sample fit and out-of-sample generalization. By contrast, W1-Recovered-Narrow is generally less accurate, which is consistent with the larger recovery errors observed in Table \ref{tab: recovery property}. A closer inspection of the direct MSE fitting on option prices shows that, even though competitive in the average fit quality, it is less informative about robustness and tail pricing. It generally achieves a larger maximum absolute percentage error. These comparisons are in line with the results of Section \ref{subsec: calibration}. The Wasserstein objective retains a noticeable regularizing effect even after the target distribution has been recovered approximately through SVI. 
	\begin{table}[htp]
		\centering
		\begin{adjustbox}{max width=\textwidth}
			\begin{tabular}{c|l|cccc}
				\toprule
				Set &Method& $\xi_0(t)$ & $H$ & $\rho$ & $\eta$ \\
				\hline
				\multirow{3}{*}{0} & MSE-Wide& 0.0009& 0.1702& 0.0287& 0.0516\\
				& W1-Recovered-Wide& 0.0365& 0.0037&0.0111 & 0.0519\\ 
				& W1-Recovered-Narrow & 0.0713& 0.0343& 0.0090& 0.1188\\
				\hline
				\multirow{3}{*}{1} & MSE-Wide& 0.0026 & 0.1115& 0.0166& 0.0248\\
				& W1-Recovered-Wide& 0.0317& 0.0095 & 0.0175& 0.0476\\ 
				& W1-Recovered-Narrow & 0.0426& 0.0922& 0.0135& 0.0484\\
				\hline
				\multirow{3}{*}{2} & MSE-Wide&0.0013& 1.3167& 0.2212& 0.2099\\
				& W1-Recovered-Wide& 0.0151& 0.5516& 0.1388& 0.0614\\ 
				& W1-Recovered-Narrow & 0.0096& 8.1784& 0.3462& 0.3377\\
				\hline
				\multirow{3}{*}{3} & MSE-Wide& 0.0123& 0.4199& 0.1235& 0.0957\\
				& W1-Recovered-Wide& 0.0068&0.1462& 0.0392& 0.0158\\ 
				& W1-Recovered-Narrow & 0.0339& 0.1053& 0.0217& 0.1049\\
				
				\bottomrule
			\end{tabular}
		\end{adjustbox}
		\caption{Absolute percentage error between $\bmtheta_{\text{true}}$ and $\bmtheta^\ast$.} 
		\label{tab: calibration-performance-recovered-parameters}
	\end{table}
	
	The corresponding parameter errors are reported in Table \ref{tab: calibration-performance-recovered-parameters}. These results further support the same conclusion. In Cases 0, 2, and 3, W1-Recovered-Wide yields the most accurate parameter estimates among the three methods. Moreover, it consistently improves upon W1-Recovered-Narrow. This demonstrates that the improvement in distribution recovery obtained from the wide strike range leads to not only better pricing performance, but also better parameter identifiability. The results also reveal an intrinsic limitation of the inverse problem that Case 2 remains challenging for all three methods. Although W1-Recovered-Wide improves significantly over W1-Recovered-Narrow, the error of parameter $H$ remains large. This behaviour is consistent with the findings in Table \ref{tab: parameter_recovery_ape} and suggests that the extremely rough regime is intrinsically harder to identify from vanilla option information. 
	
\subsubsection*{Discussion}
The practical effectiveness of the proposed Wasserstein framework rests on two complementary ingredients: an efficient simulation scheme for the model, and a sufficiently accurate recovery of the maturity-wise risk-neutral distributions from market option data. Once the recovered distribution is sufficiently accurate, Wasserstein calibration remains more robust than direct MSE fitting both in terms of out-of-sample pricing and parameter identifiability. The SVI-based preprocessing step offers a practical route for recovering the market-implied terminal distribution of the underlying asset from discrete option prices. However, it is also subject to several intrinsic limitations. First, SVI is a static parametric representation fitted separately at each maturity and therefore recovers only the maturity-wise marginal risk-neutral distribution of $S_T$, rather than the full joint dynamics generated by the rBergomi model. In particular, it does not explicitly encode the roughness of the volatility process or the inter-maturity dependence structure. Second, it is difficult to find precise conditions on the SVI parameters to prevent arbitrage. Although it is tractable by imposing penalty terms to mitigate the static arbitrage, this will, in turn, deteriorate the recovery accuracy. Finally, as observed in previous experiments, the recovered distribution is sensitive to the available strike coverage. While the central part of the distribution can be inferred reasonably well from liquid near-the-money options, the tails are harder to identify and may depend substantially on the parametric extrapolation. As a result, the SVI-recovered target distribution should be interpreted as a practical approximation to the market-implied terminal distribution, rather than an exact reconstruction of the true rBergomi distribution.

	\section{Conclusion}\label{sec: conclusion}
	
	This work introduced a computational framework for efficient pricing and calibration in the rough Bergomi model.
	
	First, we developed a modified Sum-of-Exponentials (mSOE) Monte Carlo scheme within the class of hybrid multifactor approximations. It combines a first-step exact treatment of the singular kernel near the origin with the SOE approximation scheme of Jiang et al. \cite{jiang2017fast}, and an exact Gaussian simulation of the associated multifactor components. The resulting pricing engine achieves stable convergence and high accuracy, particularly for out-of-the-money options. The numerical comparisons with both the SOE and the hybrid scheme clarify that this modification provides a practically useful refinement for repeated pricing in calibration.
	
	Second, we proposed a calibration framework based on distributional matching via the Wasserstein-1 distance. The contribution here is not a new optimal transport methodology, but an adaptation of Wasserstein-based calibration to the rBergomi setting, along with a systematic computational study. When coupled with an efficient Monte Carlo pricing engine, this approach provides a robust alternative to pointwise mean-squared error fitting. Our numerical tests indicate the improved parameter identifiability, more stable optimization behaviour, and stronger out-of-sample performance, including on barrier options.
	
	Future work will focus on several promising directions. A full empirical assessment on real market data remains an important next step. Moreover, the distribution-matching perspective studied here may be explored for other rough and non-Markovian models, provided that efficient pricing engines are available. Finally, a systematic cross-model comparison with alternative stochastic volatility models under a unified calibration protocol is an important direction for future research.

	\section*{Acknowledgements}
	GL acknowledges the support from GRF (project number: 17317122) and the Early Career Scheme (Project number: 27301921), RGC, Hong Kong.

	\bibliographystyle{abbrv}
	\bibliography{myref}
	
\end{document}